\documentclass [12pt]{article}
\usepackage{cite}
\usepackage{graphicx}
\usepackage[small]{subfigure,epsfig}
\usepackage{indentfirst}
\usepackage{amsmath,latexsym,enumerate}
\usepackage{amssymb}
\usepackage{amsfonts}
\usepackage{array,longtable,lscape}
\usepackage{amsthm}

\usepackage{color}
\usepackage{authblk}

\theoremstyle{plain} \newtheorem{theorem}{Theorem}

\newtheorem{proposition}{Proposition}

\newtheorem{corollary}{Corollary}

\numberwithin{equation}{section} \numberwithin{lemma}{section} \numberwithin{theorem}{section} \numberwithin{proposition}{section} \numberwithin{corollary}{section}
\numberwithin{remark}{section}

\usepackage{geometry} 
\begin{document}

\title{Metrisable oscillators and (super)integrable two-dimensional metrics}

\author[1]{Jaume Gin\'e}

\author[2,3]{Dmitry I. Sinelshchikov}

\affil[1]{Departament de Matem\`atica, Universitat de Lleida, Avda. Jaume II, 69; 25001 Lleida, Catalonia, Spain}
\affil[2]{Instituto Biofisika (UPV/EHU, CSIC), University of the Basque Country, Leioa E-48940, Spain}
\affil[3]{Ikerbasque Foundation, Bilbao 48013, Spain}

\maketitle

\begin{abstract}
We consider a family of nonlinear oscillators, which is the autonomous case of the two-dimensional projective connection. We construct several classes of these oscillators that are simultaneously integrable and metrisable. This leads to families of (super)integrable two-dimensional metrics that are parametrized by arbitrary functions. In the superintegrable case we obtain an explicit expression for the unparametrized geodesics. In the integrable case we present two families of metrics with transcendental first integrals. We introduce the concept of generalized Darboux integrability in the context of both projective equations and geodesic flows. We demonstrate that the constructed integrable metrics are generalized Darboux integrable.
In addition, we establish a direct connection between relative Killing vectors and invariants of the projective vector fields that are linear in the first derivative. Finally, we compute the dimensions of the projective Lie algebra for the obtained metrics, which allows us to distinguish previously known integrable cases from new ones.
\end{abstract}

\section{Introduction}
The (super)integrability of two-dimensional geodesic flows on Riemannian manifolds has been extensively studied both in mathematical physics and differential geometry \cite{Matveev2008,Matveev2011,Bolsinov2018,Matveev2019}. However there are still open problems in this area. For example, there are open problems that are connected with the existence of polynomial and rational first integrals \cite{Matveev2011,Kozlov2014,Valent2015,Bolsinov2018,Matveev2019} or explicit expressions for the geodesics \cite{Bolsinov2018,Matveev2019}. 

Furthermore, there has been interest in constructing (super)integrable metrics with polynomial and rational first integrals (see, e.g. \cite{Kruglikov2008,Kozlov2014,Valent2015,Matveev2016,Valent2017,Bagderina2017,Agapov2021,Kruglikov2024,Kruglikov2025}). For example, in \cite{Kozlov2014} it was demonstrated that there exist metrics with rational integrals of any degree. In \cite{Kruglikov2024,Kruglikov2025,Kruglikov2026} it is proved that for any metric the space of rational integrals is finite-dimensional and the connection between relative Killing tensor fields and rational integrals is discussed. 

Recently, the connections between integrability of nonlinear oscillators and geodesic flows and vice versa \cite{GS2024,Agapov2024,Agapov2025} have been studied. In this work, we extend the results obtained in \cite{GS2024} and study the metrisability of nonlinear oscillators that can be linearized with the help of nonlocal transformations \cite{Sinelshchikov2020,Sinelshchikov2020a}. We propose a new approach for constructing (super)integrable metrics from nonlinear oscillators, which is based on finding intersections between integrability and metrisability conditions. The main idea behind this approach is that nonlocal transformations neither preserve metrisability nor Lie symmetries of the projective connection, but preserve both its form and integrability. This allows one to deform a trivially metrisable projective equation, for example a linear equation, and obtain new nontrivially metrisable and integrable projective connections.

The linearizability conditions for the two-dimensional autonomous projective connections via generalized nonlocal transformations split into three cases, each of which provides a family of integrable metrics parametrized by arbitrary functions. For the first linearizability condition we demonstrate that there is a family of superintegrable metrics with a linear and a transcendental in momenta first integrals. We provide a completed classification the projective Lie algebras of this family of metrics. As a byproduct of this result, we show that for metrisable equations obtained in \cite{Agapov2024} this Lie algebra is either $\mathfrak{sl}(2,\mathbb{R})$ or $\mathfrak{sl}(3,\mathbb{R})$. This means that these integrable metrics are either Darboux-superintegrable (i.e. they have four linearly independent second-degree polynomial first integrals, see e.g. \cite{Matveev2008}) or flat, i.e. have a constant curvature. This appears to be the first time when the complete classification of the projective Lie algebras for linearizable via nonlocal transformations oscillators is carried out.

For the second and third linearizability conditions we obtain two new integrable families of the conformal metrics with, in general, transcendental first integrals. We demonstrate that these first integrals are Darboux functions (see, e.g. \cite{Llibre_book,Zhang} for the definition) of invariants of the projective connection that are polynomial in the first derivative and of relative Killing vectors for the geodesic flow. This allows us to introduce the concept of the generalized Darboux integrability for both projective vector fields and Hamiltonian geodesic flows. In contrast to the classical Darboux integrability theory (see, e.g. \cite{Llibre_book,Zhang,Goriely}), where Darboux first integrals are formed by the polynomial invariants, we allow the invariants to be polynomials only in a subset of phase variables (see also \cite{GG2003,GG2010,GS2025}). We also demonstrate that there is a direct connection between the relative Killing vectors of geodesic flows and the invariants of projective vector fields that are linear in the first derivative.

Finally, we obtain that for particular values of the parameters, one of the above transcendental first integrals can be reduced to a rational integral of an arbitrary degree. It is also worth noting that the conditions that define these integrable conformal metrics are themselves integrable second-order differential equations. Furthermore, we prove that in the generic case of these metrics their projective Lie algebra is one-dimensional.

The rest of this work is organized as follows. In the next Section we provide basic definitions that concern integrability of two-dimensional Riemannian metrics and nonlinear oscillators. We also provide direct formulas that connect the relative Killing vectors and linear in the first derivative invariants of projective vector fields. In Section \ref{sec:symm} we discuss the family of linearizable via nonlocal transformations cubic oscillators, construct their autonomous and non-autonomous first integrals and classify their Lie symmetries. We devote Section \ref{sec:metr} to the construction of superintegrable metrics with a linear and a transcendental first integrals. In Section \ref{sec:conf} we construct integrable families of the conformal metrics and introduce the notion of generalized Darboux integrability. In the last Section we briefly summarize and discuss our results.

\section{Projective equation and its integrability}
\label{sec:intro}

We begin by briefly introducing the definitions and concepts that are used in the subsequent sections.

Consider a smooth two-dimensional manifold $M$ with the (pseudo)Riemannian metric $\mathfrak{g}=(g_{ij})$, $i,j=1,2$, $g_{ij}=g_{ji}$. 
The parametrized geodesics of $\mathfrak{g}$ on $TM$ are solutions of
\begin{equation}\label{eq:geodesics_eq}
\begin{gathered}
  \ddot{x}^{i}+\Gamma^{i}_{jk}\dot{x}^{j}\dot{x}^{k}=0, \quad (x^{1},x^{2})=(x,y),
  \end{gathered}
\end{equation}
where $\Gamma^{i}_{jk}$ is the Levi-Civita connection \cite{DNF1} for the metric $\mathfrak{g}$. Notice that here and below the summation convention over repeated indices is assumed. We also denote by $\mathcal{X}_{\Gamma}=v^{i}\partial_{x^{i}}-\Gamma^{i}_{jk}v^{j}v^{k}\partial_{v^{i}}$, $v^{i}=\dot{x}^{i}$ the geodesic spray.

On $T^{*}M$ system \eqref{eq:geodesics_eq} can be presented in the Hamiltonian form:
\begin{equation}\label{eq:geodesics_H}
  \dot{x}^{i}=H_{p_{i}}, \quad \dot{p}^{i}=-H_{x_{i}}, \quad H=\frac{1}{2}g^{ij}p_{i}p_{j}, \quad g^{ik}g_{kj}=\delta^{i}_{j}.
\end{equation}
One can project \eqref{eq:geodesics_eq} on the $(x,y)$ plane to obtain
\begin{equation}\label{eq:cubic_eq}
  y_{xx}+a_{3}(x,y)y_{x}^{3}+a_{2}(x,y)y_{x}^{2}+a_{1}(x,y)y_{x}+a_{0}(x,y)=0,
\end{equation}
where
\begin{equation}\label{eq:projection}
  a_{3}=-\Gamma_{22}^{1}, \quad a_{2}=\Gamma_{22}^{2}-2\Gamma_{12}^{1}, \quad a_{1}=2\Gamma_{12}^{2}-\Gamma_{11}^{1}, \quad a_{0}=\Gamma_{11}^{2}.
\end{equation}
Therefore, for every metric there exists an equation from \eqref{eq:cubic_eq} that is a projection of its geodesic flow. However, the converse of this statement is not true, i.e. not every equation from \eqref{eq:cubic_eq} corresponds to some metric. This results from the fact that there are four coefficients of the projective connection \eqref{eq:cubic_eq} and only three components of the metric $\mathfrak{g}$. Equations from \eqref{eq:cubic_eq} that correspond to a metric are called metrisable and necessary and sufficient conditions for metrisability of \eqref{eq:cubic_eq} were obtained in \cite{Dunajski2009}. Notice that if \eqref{eq:cubic_eq} is metrisable then its solutions are unparameterized geodesics of the corresponding metric. The vector filed associated to \eqref{eq:cubic_eq} is denoted $\mathcal{X}_{p}=\partial_{x}+u\partial y -\sum_{j=0}^{3}(a_{j}u^{j})\partial_{u}$, $u=y_{x}$.

In this work we study metrisability of a family of cubic oscillators, which is a particular case of \eqref{eq:cubic_eq} and is given by
\begin{equation}\label{eq:cubic_osc}
  \begin{gathered}
 y_{xx}+k(y)y_{x}^{3}+h(y)y_{x}^{2}+f(y)y_{x}+g(y)=0.
  \end{gathered}
\end{equation}
Here $k(y)$, $h(y)$, $f(y)$ and $g(y)$ are some sufficiently smooth functions. Throughout this work we assume that $g(y)\neq0$. The following vector filed  $\mathcal{X}=u\partial_{y}-(k u^{3}+h u^{2}+f u+g)\partial_{u}$, which is $x$ independent case of $\mathcal{X}_{p}$, is associated to \eqref{eq:cubic_osc}. We also remark that the restriction to the autonomous case of the projective connection does not imply that the corresponding metric is independent of $x$.

While the metrisability conditions for the general case of \eqref{eq:cubic_osc} are cumbersome and follow from the results of \cite{Dunajski2009}, it is worth considering particular cases of \eqref{eq:cubic_osc} that admit an explicit (possibly nonautonomous) first integral. Such integral can be lifted to a first integral of \eqref{eq:geodesics_H}. If it is functionally independent of the Hamiltonian, this establishes the (super)integrability of \eqref{eq:geodesics_H}. Moreover, this idea can also be used in the reverse order, when first integrals of \eqref{eq:geodesics_H} are projected to demonstrate the integrability of an equation from \eqref{eq:cubic_osc}. Consequently, finding the intersections of integrability and metrisability conditions for \eqref{eq:cubic_osc} in particular, or more generally for \eqref{eq:cubic_eq}, provides a mechanism for generating (super)integrable metrics. Conversely, one can project the integrals of \eqref{eq:geodesics_H} onto the $(x,y)$ plane to prove the integrability of an equation from \eqref{eq:cubic_eq} or \eqref{eq:cubic_osc}.

Below, to study the metrisability of integrable cases of \eqref{eq:cubic_osc} we transform \eqref{eq:projection} into the Liouville system \cite{Liouville,Matveev2008,Dunajski2009}. Using the Christoffel symbols (see, e.g. \cite{DNF1}) to connect \eqref{eq:projection} to the metric components, and then applying the transformations 
\begin{equation}\label{eq:Liouville_1}
  \psi_{1}=\Delta^{2} g_{11}, \quad \psi_{2}=\Delta^{2}g_{12}, \quad \psi_{3}=\Delta^{2}g_{22}, \quad \Delta=\psi_{1}\psi_{3}-\psi_{2}^{2}\neq0,
\end{equation}
we obtain
\begin{equation}\label{eq:Liouville}
\begin{gathered}
  \psi_{1,x}=-\frac{2}{3}a_{1}\psi_{1}+2a_{0}\psi_{2},\\
  \psi_{3,y}=-2a_{3}\psi_{2}+\frac{2}{3}a_{2}\psi_{3},\\
  \psi_{1,y}+2\psi_{2,x}=-\frac{4}{3}a_{2}\psi_{1}+\frac{2}{3}a_{1}\psi_{2}+2a_{0}\psi_{3},\\
  \psi_{3,x}+2\psi_{2,y}=-2a_{3}\psi_{1}+\frac{4}{3}a_{1}\psi_{3}-\frac{2}{3}a_{2}\psi_{2}.
  \end{gathered}
\end{equation}
This is a linear overdetermined system of four equations for three functions $\psi_{m}$, $m=1,2,3$. An equation from \eqref{eq:cubic_eq} is metrisable if and only if there is a solution of \eqref{eq:Liouville} such that $\Delta\neq0$ \cite{Liouville,Matveev2008,Dunajski2009,Dunajski2018}.

Following \cite{Matveev2008,Matveev2012} we denote the Lie algebra of point symmetries of \eqref{eq:cubic_osc} as $\mathfrak{p}(c)$ in the general case and $\mathfrak{p}(\mathfrak{g})$ when the corresponding projective structure is metrisable. This Lie algebra is called the projective Lie algebra. It is well known \cite{Tresse,Kruglikov2009,Yumaguzhin2010,Bagderina2016} that for a general projective structure \eqref{eq:cubic_eq} the dimension of $\mathfrak{p}(c)$ belongs to $\{0,1,2,3,8\}$. Further we compute the dimensions of $\mathfrak{p}(\mathfrak{g})$ for metrisable cases of \eqref{eq:cubic_osc}. Recall also that if $\dim{\mathfrak{p}(g)}=3$ then the metric is Darboux superintegrable and if $\dim{\mathfrak{p}(g)}=8$ then the metric is flat and the corresponding projective connection can be linearized via the point transformations \cite{Matveev2008,Matveev2012}. 

Let us introduce the basic definitions that concern the integrability of the geodesic flow and its projection. A smooth function  $J=J(x,y,v^{1},v^{2})$ that satisfies the relation  $\mathcal{X}_{\Gamma}J=NJ$ for some $N=N(x,y,v^{1},v^{2})$, which is called the cofactor, is an invariant of the geodesic spray. We can introduce the invariants of the general projective connection or its autonomous case in the same way as solutions of $\mathcal{X}_{p}R=lR$, $R=R(x,y,u)$ for some cofactor $l=l(x,y,u)$ or  $\mathcal{X}R=lR$, $R=R(y,u)$ for some cofactor $l=l(y,u)$. A first integral is an invariant with zero cofactor. For \eqref{eq:cubic_osc} we also use the notion of an non-autonomous first integral which is a smooth function satisfying $\mathcal{R}_{x}+\mathcal{X}\mathcal{R}=0$ and an integrating factor $M(y,u)$, which is a smooth function that is a solution of the equation $\mbox{div}(M\mathcal{X})=0$. If there is a global, possibly except on a set of zero Lebesgue measure, explicit expression for the first integral or an integrating factor of autonomous projective equation \eqref{eq:cubic_osc}, the corresponding two-dimensional dynamical system is called completely integrable \cite{Goriely,Zhang,Llibre_book}.

Let $\{\cdot,\cdot\}$ be the canonical Poisson bracket. Then, a smooth function $T(x,y,p_{1},p_{2})$ is a first integral of \eqref{eq:geodesics_H} if $\{T,H\}=\mathcal{X}_{H}T=0$. Hamiltonian system \eqref{eq:geodesics_H} is integrable in the Liouville-Arnold sense if, apart from the Hamiltonian, there is an additional first integral, functionally independent of $H$, which is in involution with $H$. Hamiltonian system \eqref{eq:geodesics_H} is superintegrable if there are two first integrals of \eqref{eq:geodesics_H} that are in involution with $H$ and pairwise functionally independent. First integrals of the projective equation \eqref{eq:cubic_osc} are related to the first integrals of the geodesic flow \eqref{eq:geodesics_H} via the lift $y_{x}=u=H_{p_{2}}/H_{p_{1}}$. The inverse statement is true for the first integrals of \eqref{eq:geodesics_H} that are homogeneous functions of the momenta.

Recently, in \cite{Kruglikov2024,Kruglikov2025} the concept of the relative Killing tensors was introduced to study rational first integrals of geodesic flows (see also \cite{MP2004,Aoki2016}). A polynomial in momenta function of degree $d$ is called a relative Killing tensor if $ \{H,K\}=LK$, where $L$ is a necessary linear in momenta function. Geometrically this means that the Hamiltonian vector field is tangent to the submanifold $K=0$ on $T^{*}M$. 

Below we demonstrate that the first integrals obtained in Section 5 are transcendental functions of the homogeneous relative Killing tensors of degree 1 or homogeneous relative Killing vectors. Such integrals generalize functional class of rational first integrals to the class of the generalized Darboux functions. 

Finally, we establish a direct correspondence between homogeneous relative Killing vectors and polynomial with respect to $y_{x}=u$ invariants of the projective vector fields. 

\begin{proposition}
\label{p:p_KV}
Suppose that the coefficients of \eqref{eq:cubic_eq} satisfy \eqref{eq:projection} for some metric $\mathfrak{g}$ and also equation \eqref{eq:cubic_eq} has an invariant $Z=\sum_{i=1}^{2}e_{i}(x,y)u^{i-1}$, $e_{i}\in C^{\infty}(M)$ with the cofactor $l=\sum_{m=1}^{3}b_{m}(x,y)u^{m-1}$, $b_{m}\in C^{\infty}(M)$. Then the geodesic spray $\mathcal{X}_{\Gamma}$ and Hamiltonian \eqref{eq:geodesics_H} both have a homogeneous, linear in velocities, invariant and a homogeneous relative Killing vector, respectively, which have the form
\begin{align}
 J&=e_{i}v^{i},  &N&=b_{i}v^{i}-\Gamma^{1}_{1i}v^{i}-\Gamma^{1}_{12}v^{2}, \label{eq:GS_vector}\\
 K&=e^{k}p_{k},  &L&=\Gamma^{1}_{1j}g^{jk}p_{k}+\Gamma^{1}_{12}g^{2k}p_{k}-b^{k}p_{k}, \quad  e^{k}=g^{kj}e_{j}, \quad b^{k}=g^{kj}b_{j}.  \label{eq:Killing_vector}
\end{align}
Inversely, if  Hamiltonian system \eqref{eq:geodesics_H} has a homogeneous relative Killing vector, then equation \eqref{eq:cubic_eq}  whose coefficients satisfy \eqref{eq:projection} has an invariant that is linear in $y_{x}=u$  with a quadratic in $y_{x}=u$ cofactor.
\end{proposition}
\begin{proof}
Suppose that coefficients of \eqref{eq:cubic_eq} satisfy \eqref{eq:projection} for some metric $\mathfrak{g}$ and also equation \eqref{eq:cubic_eq} has an invariant $Z=\sum_{i=1}^{2}e_{i}(x,y)u^{i-1}$ with the cofactor $l=\sum_{m=1}^{3}b_{m}(x,y)u^{m-1}$. Substituting the expression for the invariant and the cofactor in the respective definition and collecting coefficients at the same powers of $u$ one can find that $b_{3}=\Gamma^{1}_{22}$.

Consider $\tilde{J}=Z|_{z=\frac{v}{u}}$ and $\tilde{N}=v^{1}l|_{z=\frac{v^{2}}{v^{1}}}$. Then, compute $X_{\Gamma}\tilde{J}=\tilde{N}\tilde{J}$ and $XR=lR$ collecting coefficients at the powers of $v^{j}$ and $u$, respectively. Comparing the equations for $e_{i}$ and $b_{m}$ obtained in this way one can demonstrate that if $XR=lR$ is satisfied, then so is $X_{\Gamma}\tilde{J}=\tilde{N}\tilde{J}$ and vice versa. 

Both $\tilde{J}$ and $\tilde{N}$ are rational functions of $v^{1}$. Note also that $v^{1}$ is the invariant of $X_{\Gamma}$ with the cofactor $-(\Gamma^{1}_{jk}v^{j}v^{k})/v^{1}$. Then consider the invariant $J=\tilde{J}v^{1}=e_{i}v^{i}$, which is a linear homogeneous polynomial in $v^{i}$, $i=1,2$. The corresponding cofactor has the form $N=b_{i}v^{i}-\Gamma^{1}_{1i}v^{i}-\Gamma^{1}_{12}v^{2}$ (recall that $b_{3}=\Gamma^{1}_{22}$). As a consequence, we have that the geodesic spray $X_{\Gamma}$ has a linear homogeneous invariant in velocities with a linear homogeneous in velocities cofactor.

Finally, it is known that the Legendre transformations $v^{j}=g^{jk}p_{k}$ map the geodesic spray into the Hamiltonian flow $\{\cdot,H\}$ (see, e.g. \cite{Arnold_book}). Recall also that in the definition of the relative Killing vector the flow $\{H,\cdot\}=-\{\cdot,H\}$ is used and, consequently, the cofactor changes the sign. As a result, we obtain expressions \eqref{eq:Killing_vector}. 

The consideration above can be repeated in the reverse order to demonstrate that for any Killing vector of \eqref{eq:geodesics_H} there is an invariant that is linear in $u=y_{x}$ of the metrisable projective structure. This completes the proof.
\end{proof}

Proposition \ref{p:p_KV} is used in Section \ref{sec:conf} to map invariants of \eqref{eq:cubic_osc}, which are linear in $u=y_{x}$, to the relative Killing vectors of integrable Hamiltonian geodesic flows \eqref{eq:geodesics_H}.

\section{Cubic oscillators linearizable via nonlocal transformations and projective symmetries}
\label{sec:symm}

In this section we briefly present the results on the linearizability of family \eqref{eq:cubic_osc} with respect to nonlocal transformations obtained in \cite{Sinelshchikov2020,Sinelshchikov2020a}. We demonstrate that the linearizability is equivalent to the existence of a cubic rational, with respect to $u$, integrating factor and a  first integral, and, in  a particular case, a non-autonomous first integrals.

As a linear equation we consider
\begin{equation}\label{eq:linear_eq}
  w_{\xi\xi}+\beta w_{\xi}+\alpha w=0.
\end{equation}
Here $\alpha\beta\neq0$ are arbitrary parameters. 

We use the following family of nonlocal transformations
\begin{equation}\label{eq:GNT}
  w=F(y), \quad d\xi=[G_{1}(y)y_{x}+G_{2}(y)]dx,
\end{equation}
where it is assumed that $F$, $G_{1}$ and $G_{2}$ are sufficiently smooth functions and $F_{y}G_{1}G_{2}\neq0$, since otherwise transformations \eqref{eq:GNT} either degenerate or are reduced to point or generalized Sundman ones.

Transformations \eqref{eq:GNT} are a combination of point, contact, and Sundman transformations. They preserve autonomous projective structures, i.e., the family of equations \eqref{eq:cubic_osc} is closed under \eqref{eq:GNT}. If $G_{2}=0$ then \eqref{eq:GNT} reduce to point transformations. If $G_{1}=0$, then \eqref{eq:GNT} are generalized Sundman transformations and if, in addition, $F(y)=y$ then \eqref{eq:GNT} are Sundman transformations. It is clear that point transformations preserve both the metrisability and the projective Lie algebra of \eqref{eq:cubic_osc}. However, the Sundman transformations, as a special case, and transformations \eqref{eq:GNT} in general, neither preserve the metrisability of \eqref{eq:cubic_osc} nor its projective Lie algebra (see, e.g. \cite{Carinena2022}). 

This allows us to propose the following approach for constructing integrable projective structures. Suppose we start from a trivially metrisable and integrable projective vector field, for instance, a linear equation, and construct its equivalence class with respect to transformations \eqref{eq:GNT}. Transformations \eqref{eq:GNT} preserve the integrability of  \eqref{eq:cubic_osc}, since they map autonomous first integrals to autonomous first integrals \cite{Sinelshchikov2021}. However, metrisability is not preserved, and thus we need to find the intersection between the equivalence class and metrisability conditions. If this intersection is not empty, we obtain families of projective structures that are both integrable and metrisable. Finally, we lift the corresponding first integral and/or invariant of the projective connection to a first integral and/or invariant of the geodesic flow. Here. we apply this approach to linear equation \eqref{eq:linear_eq} and to the autonomous case of the projective connection \eqref{eq:cubic_osc}. We plan to consider the general case of the projective connections that are linearizable via nonlocal transformations elsewhere.

Further we need the following simple result:

\begin{proposition}
  Linear equation \eqref{eq:linear_eq} possesses two first order polynomial invariant curves
\begin{equation}\label{eq:linear_IC}
  \widetilde{Z}_{1,2}=2v+(\beta\pm\sqrt{\beta^{2}-4\alpha})w, \quad \lambda_{1,2}=-\frac{\beta\mp\sqrt{\beta^{2}-4\alpha}}{2}.
\end{equation}
If, in addition the following relation on the parameters $\alpha$, $\beta$ holds
\begin{equation}\label{eq:alpha_def}
2\beta^{2} -9\alpha=0,
\end{equation}
then \eqref{eq:linear_eq} has a second order polynomial invariant curve
\begin{equation}\label{eq:linear_IC_1}
  \widetilde{Z}_{3}=c_{1}\widetilde{Z}_{1}^{2}+c_{2}\widetilde{Z}_{2}, \quad \lambda_{3}=-\frac{2\beta}{3}, \quad c_{1}c_{2}\neq0.
\end{equation}
The integrating factor of \eqref{eq:linear_eq} corresponding to \eqref{eq:linear_IC} has the form
\begin{equation}\label{eq:linear_M}
  \widetilde{M}=(\widetilde{Z}_{1}\widetilde{Z}_{2})^{-1},
\end{equation}
while the first integral is
\begin{equation}\label{eq:linear_FI}
  \widetilde{R}=\widetilde{Z}_{1}^{\sqrt{\beta^{2}-4\alpha}+\beta}\widetilde{Z}_{2}^{\sqrt{\beta^{2}-4\alpha}-\beta}.
\end{equation}
We use the notation $v=w_{\xi}$ in all of the above expressions.
\end{proposition}
\begin{proof}
It is not difficult to demonstrate that the highest degree of an invariant curve of \eqref{eq:linear_eq} is 2. It is also clear that the cofactor is a constant, since the degree of the vector field associated to \eqref{eq:linear_eq} is 1. Thus, it is straightforward to verify that the only possible first order invariant curves of \eqref{eq:linear_eq} are \eqref{eq:linear_IC} and an irreducible second order invariant exists if and only if \eqref{eq:alpha_def} holds and has the from \eqref{eq:linear_IC_1}. Integrating factor \eqref{eq:linear_M} and first integral \eqref{eq:linear_FI} can be directly found from  invariant curves \eqref{eq:linear_IC} This completes the proof.
\end{proof}

Now we present equivalence criterion for \eqref{eq:cubic_eq} and linear equation \eqref{eq:linear_eq} via transformations \eqref{eq:GNT}. Below we use the following notations to simplify the presentation of the results
\begin{equation}\label{eq:cubic_notation}
  A=27kg^2-9hfg  + 2f^3 + 9gf_{y}- 9fg_{y}, \quad B=3gh-f^2 + 3g_{y}.
\end{equation}

The following proposition holds (see Theorem 2.1 and Corollary 2.1 in \cite{Sinelshchikov2020}):
\begin{proposition}
\label{p:p2}
  Suppose that $g\neq0$. Then the following statements are equivalent:
   \begin{enumerate}
    \item[1)] equation \eqref{eq:cubic_osc} can be transformed into \eqref{eq:linear_eq} via \eqref{eq:GNT}, where the functions $F$ and $G$ are given by
    \begin{equation}\label{eq:cubic_lc_tr_gen}
      F_{y}=\frac{\alpha F G_{2}^{2}}{g}, \quad G_{1}=\frac{G_{2}(f-\beta G_{2})}{3g},
    \end{equation}
    and $G_{2}$ satisfies by one of the following correlations
     \begin{itemize}
                 \item[1a)]
                 \begin{equation}\label{eq:cubic_lc_tr_G2_1}
                 G_{2}=\frac{(\beta^2 - 3\alpha)A^{2}}{\beta(2\beta^2 - 9\alpha)(gA_{y}-BA)}, \quad A(\beta^2 - 3\alpha)( 2\beta^2 - 9\alpha)(gA_{y}-BA)\neq0;
                 \end{equation}
                 \item[1b)]
                  \begin{equation}
                 \label{eq:cubic_lc_tr_G2_2}
                  G_{2}^{3}=-\frac{A}{\beta^{3}}, \quad  \beta^2 - 3\alpha=0, \quad A\neq0;
                        \end{equation}
                 \item [1c)]
                     \begin{equation}  \label{eq:cubic_lc_tr_G2_3}
                    G_{2,y}=\frac{G_{2}(\beta^{2}G_{2}^{2}+3B)}{9g},\quad 2\beta^2 - 9\alpha=0,
                           \end{equation}

               \end{itemize}

    \item[2)] the functions $k$, $h$, $f$, $g$ and the parameters $\alpha$ and $\beta$ satisfy one of the relations
               \begin{itemize}
                 \item[2a)]
                 \begin{equation}\label{eq:cubic_lc_1}
                 \begin{gathered}
                    \beta^2(2\beta^2 - 9\alpha)^2(g A_{y} - BA)^3 - A^{5}(\beta^2 - 3\alpha)^{3}=0,\\
                   (2\beta^2 - 9\alpha) (\beta^2 - 3\alpha)A(gA_{y}-BA)\neq0;
                    \end{gathered}
                 \end{equation}
                 \item[2b)]
                   \begin{gather}
                    \beta^2 - 3\alpha=0,\label{eq:cubic_lc_2_p} \\
                    g A_{y} - BA=0 \label{eq:cubic_lc_2};
                 \end{gather}
                 \item [2c)]
                    \begin{gather}
                    2\beta^2 - 9\alpha=0, \label{eq:cubic_lc_3_p} \\
                    A=0  \label{eq:cubic_lc_3};
                 \end{gather}

               \end{itemize}

    \item[3)] equation \eqref{eq:cubic_osc} has a rational and cubic with respect to $u=y_{x}$ integrating factor
  \begin{equation}\label{eq:lc_M}
   M=\frac{gG_{2}^{2}}{[(f-\beta G_{2})u+3g][(2\beta-9\alpha^{2})G_{2}^{2}u^{2}-\beta G_{2}(fu+3g)u-(fu+3g)^{2}] };
  \end{equation}
    \item[4)] equation \eqref{eq:cubic_osc} possesses a first integral
 \begin{multline}\label{eq:lc_FI}
R=F^{2\rho}[(f-\beta G_{2})u+3g]^{-2\rho}\left(6\alpha G_{2}u+(\beta +\rho) [(f-\beta G_{2})u+3g]\right)^{\rho+\beta} \times \\  \left(6\alpha G_{2}u+(\beta -\rho) [(f-\beta G_{2})u+3g]\right)^{\rho-\beta}, \quad \rho=\sqrt{\beta^{2}-4\alpha}
  \end{multline}
  \end{enumerate}
\end{proposition}

It is also worth noting that correlation \eqref{eq:cubic_lc_tr_G2_3} is the Bernoulli differential equation for the function $G_{2}$, which can be explicitly solved. We do not provide the corresponding expression in order to simplify the presentation of the results.

In what follows, we also use the notation
\begin{equation}\label{eq:C_def}
  B=\frac{3}{2} \left(g_{y}+\frac{gC_{yy}}{C_{y}}\right), \quad C_{y}\neq0,
\end{equation}
which, in terms of the functions $h$, $f$ and $g$ is equivalent to
\begin{equation}\label{eq:C_sol}
  C=\int g(y)\exp\left\{2\int \left(h-\frac{f^{2}}{3g}\right)dy\right\}dy.
\end{equation}

The following statement holds:
\begin{corollary}
If condition \eqref{eq:alpha_def} holds, then linearizable via \eqref{eq:GNT} equations have a rational cubic integrating factor
\begin{equation}\label{eq:cubic_M_3}
M=\frac{gG_{2}^{2}}{(fu+3g)([\beta G_{2}+f]u+3g)([\beta G_{2}-f]u-3g)},
\end{equation}
rational quadratic with respect to $u$ first integral
\begin{equation}\label{eq:cubic_FI_3}
R=\frac{9gC_{y}u^{2}}{(fu+3g)^{2}}+2C,
\end{equation}
and the non-autonomous  first integral
\begin{equation}\label{eq:cubic_Inv_3}
\mathcal{R}=\int\limits_{0}^{y}\left(\frac{f(\omega)}{3g(\omega)}-\sqrt{\frac{C^{'}(\xi)}{g(\omega)[R-2C(\omega)]}}\right)d\omega+x,
\end{equation}
where $R$ is given in \eqref{eq:cubic_FI_3} and $\omega$ is an integration variable.
\end{corollary}

Now we obtain a classification of possible dimensions of Lie algebras admitted by \eqref{eq:cubic_osc} when condition \eqref{eq:cubic_lc_3} is satisfied. This classification is used below to find cases of the Darboux superintegrable and flat metrics that correspond to \eqref{eq:cubic_osc} under \eqref{eq:cubic_lc_3}. To this aim, we use the results of classification of Lie point symmetries of the projective equations given in \cite{Bagderina2016}. In \cite{Bagderina2016} all equations of the from \eqref{eq:cubic_eq} are classified in terms of their relative and basis differential invariants of point transformations. The dimension of the admitted Lie algebra can be determined by basis algebraic invariants. Below we compute relative and absolute invariants for \eqref{eq:cubic_eq} when condition \eqref{eq:cubic_lc_3} is met.

Recall also that projective equation \eqref{eq:cubic_eq} can be linearized via point transformations and $\dim{\mathfrak{p}(c)}=8$ if and only if the Liouville invariants $\mathcal{L}_{1}$ and $\mathcal{L}_{2}$ vanish simultaneously \cite{Liouville,Tresse,Kruglikov2009,Yumaguzhin2010}. For equation \eqref{eq:cubic_osc} the Liouville invariants have a very compact representation in terms of the functions $A$ and $B$ defined in \eqref{eq:cubic_notation}:
\begin{equation}\label{eq:Liouville_inv_co_AB}
  \mathcal{L}_{1}=-\frac{B_{y}}{3}, \quad \mathcal{L}_{2}=\frac{A_{y}-3fB_{y}}{27g},
\end{equation}
which is used in the proof of the following result:

\begin{theorem}
\label{th:thc1}
Suppose that correlation \eqref{eq:cubic_lc_3} holds and $g\neq0$. Then $\dim{\mathfrak{p}(c)}$ can be one, two, three or eight. In the generic case $\dim{\mathfrak{p}(c)}=1$.   We have that $\dim{\mathfrak{p}(c)}=2$ if and only if
\begin{equation}\label{eq:cubic_lc_3_2DLA}
  gBB_{yy} + g_{y}BB_{y} - 2gB_{y}^{2}=0, \quad  g B_{y}(9gB_{y}-4B^{2})\neq0.
\end{equation}

We have that $\dim{\mathfrak{p}(c)}=3$ and $\mathfrak{p}(c)$  is isomorphic to $\mathfrak{sl}(2,\mathbb{R})$  if and only if
  \begin{gather}\label{eq:cubic_lc_3_3DLA}
  9 g^{2} B_{yy}^{2}+9 g \left(B +2 g_{y}\right) B_{y} B_{yy}-25 g B_{y}^{3}-\left(4 B +3 g_{y}\right) \left(B -3 g_{y}\right) B_{y}^{2}=0, \quad g B_{y}\neq0,\\
  2gB_{yy} + BB_{y} + 2g_{y}B_{y}\neq0.
\end{gather}
Finally, $\dim{\mathfrak{p}(c)}=8$ and $\mathfrak{p}(c)$  is isomorphic to $\mathfrak{sl}(3,\mathbb{R})$ if and only if
  \begin{equation}\label{eq:cubic_lc_3_8DLA}
B_{y}=0, \quad g \neq0.
\end{equation}
\end{theorem}
\begin{proof}
In order to obtain the proof of this statement we use Theorems 1 and 4 from \cite{Bagderina2016}. Notice that in \cite{Bagderina2016} instead of the Liouville invariants the notations $\beta_{1}=-\mathcal{L}_{1}$ and $\beta_{2}=-\mathcal{L}_{2}$ are used. We will stick to the classical notation and use expressions \eqref{eq:Liouville_inv_co_AB} for the Liouville invariants. In what follows we also express invariants and other quantities for \eqref{eq:cubic_eq} in terms of the functions $A$ and $B$ (see definitions in \eqref{eq:cubic_notation}). All other notation is inherited from \cite{Bagderina2016}. Since \eqref{eq:cubic_eq} does not depend explicitly on $x$, it is clear that $\dim{\mathfrak{p}(c)}\in \{1,2,3,8\}$. Below, we demonstrate that all four cases are possible.

We use the following criteria for different dimensions of $\mathfrak{p}(c)$. According to  \cite{Bagderina2016} (see Theorems 1 and 4 there) we have that: 1) $\dim{\mathfrak{p}(c)}=2$ if and only if all algebraic basis invariants of \eqref{eq:cubic_eq} are constant; 2) $\dim{\mathfrak{p}(c)}=3$ if and only if $J_{0}=j_{0}=j_{1}=j_{2}=\Gamma_{0}j_{3}+5\mathcal{L}_{1}=0$ and $\mathcal{L}_{1}\neq0$. It is well known that $\dim{\mathfrak{p}(c)}=8$ if and only if $\mathcal{L}_{1}=\mathcal{L}_{2}=0$.

Suppose that \eqref{eq:cubic_lc_3} holds. Then using Theorem 1 of \cite{Bagderina2016} we compute the quantities $J_{0}$, $j_{k}$, $k=0,1,2,3$, $\Gamma_{0}$, $\Lambda$ and $e_{0}$ for \eqref{eq:cubic_eq}. We also substitute condition \eqref{eq:cubic_lc_3} into the definition of the Liouville invariants \eqref{eq:Liouville_inv_co_AB} to obtain expressions for $\mathcal{L}_{1}$  and $\mathcal{L}_{2}$. As a result, we find that the relative invariant $J_{0}$ and the functions $j_{0}$, $j_{2}$ and $e_{0}$ vanish. The rest of the above quantities are given by
\begin{align}\label{eq:thc1_eq1}
  j_{1}  =&-\frac{1}{486g^{2}}\bigg(45g^{2} B_{y}B_{y yy}-54 g^{2} B_{yy}^{2} -9 g B_{y} \left(B -3 g_{y}\right) B_{yy}-15 B_{y}^{3} g + \nonumber \\
         &\left(4 B^{2}-9 B g_{y}+45 g g_{yy}-9  g_{y}^{2}\right) B_{y}^{2} \bigg),\nonumber \\
  j_{3} =& -\frac{9 \left(4 B B_{y}+3 g B_{y y}+3 B_{y} g_{y}\right)}{5 B_{y}^{3}}, \quad
  \Gamma_{0}= \frac{B_{y} \left(3 g B_{y y}+3 B_{y} g_{y}-B B_{y}\right)}{27 g},\\
  \Lambda=&\frac{9}{25B_{y}^{6}}\bigg(135 B_{y} B_{yyy} g^{2}+252 B_{yy}^{2} g^{2}+3 B_{y} g \left(374 B +303 g_{y}\right) B_{y ,y}+ \nonumber \\
          &B_{y}^{2} \left(808 B^{2}+1122 B g_{y}+30 g B_{y}+135 g g_{y ,y}+387 g_{y}^{2}\right)\bigg), \nonumber \\
 \mathcal{L}_{1}=& \frac{B_{y}}{3}, \quad \mathcal{L}_{2}=\frac{fB_{y}}{9g}. \nonumber
\end{align}

We see that $\mathcal{L}_{1} j_{1}\neq0$ for arbitrary functions $f$, $g$, $h$ and $k$ satisfying \eqref{eq:cubic_lc_3}. Therefore, a generic equation from \eqref{eq:cubic_eq} satisfying \eqref{eq:cubic_lc_3} falls into the fourth class of equations listed in \cite{Bagderina2016}. Then, taking into account that $j_{2}=0$, we compute basis algebraic invariants
\begin{equation}\label{eq:thc1_eq2}
 I_{1}=-\frac{\Gamma_{0}}{\mathcal{L}_{1}j_{1}^{1/2}}, \quad I_{2}=10 j_{1}^{1/2}j_{3}, \quad \mathcal{L}_{1} j_{1}\neq0.
\end{equation}
As a result, we find that $I_{1}$ and $I_{2}$ are not constant for arbitrary choice of functions $f$, $g$, $h$ and $k$ satisfying \eqref{eq:cubic_lc_3}. Therefore, in the generic case, the Lie algebra of point symmetries of \eqref{eq:cubic_eq} satisfying \eqref{eq:cubic_lc_3} is one dimensional. Let us also remark that when $j_{1}=0$ and $j_{3}\mathcal{L}_{1}(\Gamma_{0}j_{3}-5\beta_{1})\neq0$, an equation from \eqref{eq:cubic_eq} whose coefficients fulfill \eqref{eq:cubic_lc_3} falls into the sixth class of equations listed in \cite{Bagderina2016}. The corresponding basis relative invariants are
\begin{equation}\label{eq:thc1_eq3}
 I_{1}=-\frac{\Gamma_{0}j_{3}}{\mathcal{L}_{1}}, \quad I_{2}=\frac{(5-I_{1})\Lambda}{5j_{3}^{2}}, \quad \mathcal{L}_{1} j_{3}\left(\Gamma_{0}j_{3}+5\mathcal{L}_{1}\right)\neq0,
\end{equation}
where we use the fact that $e_{0}=0$. The invariants given in \eqref{eq:thc1_eq3} are also non-constant for an arbitrary choice of unctions $f$, $g$, $h$ and $k$ satisfying \eqref{eq:cubic_lc_3}. Consequently, we see again that generically $\dim{\mathfrak{p}(c)}=1$.

Let us classify equations from \eqref{eq:cubic_eq} that satisfy \eqref{eq:cubic_lc_3} and have that $\dim{\mathfrak{p}(c)}=2$. First, we suppose that $\mathcal{L}_{1}j_{1}\neq0$ and require that $I_{1,y}=I_{2,y}=0$ for $I_{1}$, $I_{2}$ provided in \eqref{eq:thc1_eq2}. As a result, we obtain the following system of equations
\begin{equation}
\begin{gathered}\label{eq:thc1_eq4}
 \mathcal{L}_{1}j_{1}\Gamma_{0,y}-\Gamma_{0}j_{1}\mathcal{L}_{1,y}-\frac{1}{2}\Gamma_{0}\mathcal{L}_{1}j_{1,y}=0,\\
 j_{3}j_{1,y}+2j_{1}j_{3,y}=0.
\end{gathered}
\end{equation}
Assume that $j_{3}\neq0$. Then, excluding $j_{1}$ from the first equation of \eqref{eq:thc1_eq4} via the second one we get
\begin{equation}\label{eq:thc1_eq5}
j_{3}\mathcal{L}_{1}\Gamma_{0,y}-\Gamma_{0}j_{3}\mathcal{L}_{1,y}+\Gamma_{0}\mathcal{L}_{1}j_{3,y}=0.
\end{equation}
If we substitute the expressions for $j_{3}$, $\mathcal{L}_{1}$ and $\Gamma_{0}$ into \eqref{eq:thc1_eq5}, we obtain a third-order differential equation with respect to $B$. The coefficient in front of $B_{yyy}$ has the form
\begin{equation}\label{eq:thc1_eq6}
C_{3}=2gB_{yy} + BB_{y} + 2g_{y}B_{y}.
\end{equation}
We suppose that $C_{3}\neq0$ and substitute the expression for $B_{yyy}$ obtained from \eqref{eq:thc1_eq5} into the second equation of \eqref{eq:thc1_eq4}. As a consequence, we get
\begin{equation}\label{eq:thc1_eq7}
\frac{625B_{y}^{5}g^{2}\Gamma_{0}^2j_{3}(j_{3}^{2}B_{y}^3 + 9j_{3}BB_{y} - 81g)(gBB_{yy} +g_{y}BB_{y}- 2gB_{y}^{2})}{9C_{3}^{3}}=0.
\end{equation}
First, we suppose that $gBB_{yy} +g_{y}BB_{y}- 2gB_{y}^{2}=0$, which meant that the condition from \eqref{eq:cubic_lc_3_2DLA} holds. One can verify that under this conditions both equations in \eqref{eq:thc1_eq4} vanish. Moreover, integrating \eqref{eq:cubic_lc_3_2DLA} once we get
\begin{equation}\label{eq:thc1_eq8}
gB_{y}-c_{1}B^{2}=0,
\end{equation}
where $c_{1}$ is an integration constant.

Using both \eqref{eq:cubic_lc_3_2DLA} and \eqref{eq:thc1_eq8} one can demonstrate that $I_{1}$ and $I_{2}$ are constants. However, the function $j_{1}$ vanishes when $c_{1}=1/6$ or $c_{1}=4/9$, which is the same as $(9gB_{y}-4B^{2})(6gB_{y}-B^{2})=0$. These cases will be taken into account below when we consider the case of $j_{1}=0$. Moreover, from the condition $\Gamma_{0}=0$ it follows that $j_{1}=0$ and, hence, with this case we also deal below. We are only left with the possibility of $j_{3}^{2}B_{y}^3 + 9j_{3}BB_{y} - 81g=0$, which leads to $j_{1}=0$ as well. If we assume that $C_{3}=0$ we obtain that both equations in \eqref{eq:thc1_eq4} vanish if and only if $B^{2}+4gB_{y}=0$, which is a particular case of \eqref{eq:thc1_eq8}. Therefore, the case of $C_{3}=0$ is included into \eqref{eq:cubic_lc_3_2DLA}. 

Finally, one can show that $j_{3}$ and the first equation from \eqref{eq:thc1_eq4} vanish simultaneously if and only if $2B^{2}+3gB_{y}=0$, which is again a particular case of  \eqref{eq:cubic_lc_3_2DLA}. As a consequence, we obtain that all equations from \eqref{eq:cubic_osc}, whose coefficients satisfy \eqref{eq:cubic_lc_3}, have $\dim{\mathfrak{p}(c)}=2$ in the case of $j_{1}\neq0$ if and only if \eqref{eq:cubic_lc_3_2DLA} is fulfilled provided that $(9gB_{y}-4B^{2})(6gB_{y}-B^{2})\neq0$. Below we show that the condition $6gB_{y}-B^{2}\neq0$ can be removed since it corresponds to the equations from \eqref{eq:cubic_osc} under condition \eqref{eq:cubic_lc_3}, which falls into the sixth  class of equations listed in \cite{Bagderina2016} and having two-dimensional Lie algebra of point symmetries. If $9gB_{y}-4B^{2}=0$ we obtain that $j_{1}=\Gamma_{0}j_{3}+5\mathcal{L}_{1}=0$ and the corresponding equations from \eqref{eq:cubic_eq} have $\dim{\mathfrak{p}(c)}=3$ and this case is considered below separately.

In order to finish with the equations from \eqref{eq:cubic_eq} having two-dimensional Lie algebra we consider the case of $j_{1}=0$ assuming that $\mathcal{L}_{1}j_{3}(\Gamma_{0}j_{3}+5\mathcal{L}_{1})\neq0$. Otherwise, we will have that $\dim{\mathfrak{p}(c)}$ is $3$ or $8$. In the same way as in the case of $j_{1}\neq0$ we differentiate basis invariants given in \eqref{eq:thc1_eq3} and require that these derivative are zeros. The corresponding computations are similar to those of the case of $j_{1}\neq0$ and, thus, are omitted. The only condition that fulfills that $\mathcal{L}_{1}j_{3}(\Gamma_{0}j_{3}+5\mathcal{L}_{1})\neq0$ and $I_{1,y}=I_{2,y}=0$ is $6gB_{y}-B^{2}=0$. This is a particular case of \eqref{eq:thc1_eq8} and, hence, of \eqref{eq:cubic_lc_3_2DLA}. Therefore, as a criterion for \eqref{eq:cubic_eq} under condition \eqref{eq:cubic_lc_3} to have $\dim{\mathfrak{p}(c)}=2$ we can use \eqref{eq:cubic_lc_3_2DLA}, taking into account that if $6gB_{y}-B^{2}\neq0$ or $6gB_{y}-B^{2}=0$ the basis invariants are given by \eqref{eq:thc1_eq2} or  \eqref{eq:thc1_eq3}, respectively.


Now we consider the case of a three-dimensional Lie algebra of point symmetries. Taking into account that $j_{0}=j_{2}=0$ for \eqref{eq:cubic_eq} under \eqref{eq:cubic_lc_3}, we have that according to Theorem 1 of \cite{Bagderina2016} $\dim \mathfrak{p}(c)=3$ if and only if $j_{1}=\Gamma_{0}j_{3}+5\mathcal{L}_{1}=0$, provided that $\mathcal{L}_{1}\neq0$.  Condition $\Gamma_{0}j_{3}-5\beta_{1}=0$ is equivalent \eqref{eq:cubic_lc_3_3DLA}. Moreover, if we differentiate \eqref{eq:cubic_lc_3_3DLA} once, we get
\begin{multline}\label{eq:thc1_eq9}
9g^{2}B_{y}C_{3} B_{yyy}-18 g^{3} B_{y ,y}^{3}-
 g \left(16 g B_{y}-9 B^{2}-27 B g_{y}-18 g g_{y y}-18 g_{y}^{2}\right) B_{y}^{2} B_{yy}-\\
 \left(4 B^{3}-9 B^{2} g_{y}+33 B g B_{y}-9 B g g_{yy}-9 B g_{y}^{2}+16 g B_{y} g_{y}-18 g g_{y} g_{yy}\right)B_{y}^{3}=0,
\end{multline}
where $C_{3}$ is given by \eqref{eq:thc1_eq6}. If $C_{3}=0$, then $\Gamma_{0}j_{3}+5\mathcal{L}_{1}$ and $j_{1}$ vanish simultaneously if and only if $B=0$, which is a particular case of condition \eqref{eq:cubic_lc_3_8DLA} for $\dim \mathfrak{p}(c)=8$. Thus, we assume that $C_{3}\neq0$ and exclude $B_{yyy}$ from $j_{1}$ via \eqref{eq:thc1_eq9} and then use \eqref{eq:cubic_lc_3_3DLA} to exclude $B_{yy}^{2}$. As a result, we obtain that $j_{1}=0$. Therefore, $\Gamma_{0}j_{3}+5\mathcal{L}_{1}=0$ yields the condition $j_{1}=0$ if $C_{3}\neq0$. As a consequence, we have that equations from family \eqref{eq:cubic_eq} that satisfy \eqref{eq:cubic_lc_3} have a three-dimensional projective Lie algebra if and only if condition \eqref{eq:cubic_lc_3_3DLA} is fulfilled. Moreover, it is known (see, e.g. \cite{Kruglikov2009,Bagderina2016}) that this three dimensional Lie algebra is $\mathfrak{sl}(2,\mathbb{R})$.

Finally, an equation from \eqref{eq:cubic_eq} can have an eight-dimensional projective Lie algebra if and only if both Liouville invariants vanish. Using the expressions for $\mathcal{L}_{1}$ and $\mathcal{L}_{2}$ given in \eqref{eq:thc1_eq1} we arrive at the condition \eqref{eq:cubic_lc_3_8DLA}. It is known that this eight-dimensional Lie algebra of point symmetries is isomorphic to $\mathfrak{sl}(3,\mathbb{R})$. This completes the proof.
\end{proof}

\textbf{Example 1}. We provide an example of an equation from \eqref{eq:cubic_osc} that illustrates that all dimensions of $\mathfrak{p}(c)$ listed in Theorem \ref{th:thc1} are possible. Consider the oscillator
\begin{equation}\label{eq:ex1_1}
  y_{xx}+\left(\frac{1}{y}+\frac{\mu-1}{y^{2}}+\frac{\nu}{3y^{3}}\right)y_{x}^{3}+\left(3+\frac{\mu-2}{y}+\frac{\nu}{3y^{2}}\right)y_{x}^{2}+3y y_{x}+y^{2}=0,
\end{equation}
First, it is not difficult to check that coefficients of \eqref{eq:ex1_1} satisfy \eqref{eq:cubic_lc_3}. Then, using \eqref{eq:cubic_notation} we find that $B=3\mu y+\nu$. Substituting $B$ and $g$ into \eqref{eq:cubic_lc_3_2DLA} we obtain that it vanishes if and only if $\mu\nu=0$. Consequently, for arbitrary non-zero values of the parameters  $\mu$ and $\nu$ we have that $\dim{\mathfrak{p(c)}}|_{\eqref{eq:ex1_1}}=1$. It is also clear that $B_{y}=3\mu$ and, hence, if $\mu=0$, then $\dim{\mathfrak{p(c)}}|_{\eqref{eq:ex1_1}}=8$. Thus, further, we assume that $\mu\neq0$ and $\nu=0$. Substituting the expressions for $B$ and $g$ into \eqref{eq:cubic_lc_3_2DLA} we obtain that $\dim{\mathfrak{p(c)}}|_{\eqref{eq:ex1_1}}=2$ if $\mu\neq3/4$. Using relations \eqref{eq:cubic_lc_3_3DLA} we obtain that $\dim{\mathfrak{p(c)}}|_{\eqref{eq:ex1_1}}=3$ if and only if $\mu=3/4$.

\section{Metrisability of cubic oscillators}
\label{sec:metr}

In this section we classify all metrisable cubic oscillators that satisfy condition \eqref{eq:cubic_lc_3}. We also demonstrate that all metrics described in Theorem 3 of \cite{Agapov2024} are either Darboux-superintegrable or flat. We show how to obtain non-trivial superintegrable metrics for metrisable equations from \eqref{eq:cubic_eq} under correlation  \eqref{eq:cubic_lc_3}.

\begin{theorem}
\label{th:thc2}
Suppose that correlation \eqref{eq:cubic_lc_3} holds and $g\neq0$. Then equation \eqref{eq:cubic_osc} is metrisable if and only if one of the following conditions holds:
\begin{enumerate}
  \item[1)] $\psi_{1,x}=0$;
  \item[2)] correlation \eqref{eq:cubic_lc_3_3DLA} is fulfilled or $\dim{\mathfrak{p}(\mathfrak{g})}=3$;
  \item[3)] correlation \eqref{eq:cubic_lc_3_8DLA} is fulfilled or $\dim{\mathfrak{p}(\mathfrak{g})}=8$.
\end{enumerate}
\end{theorem}
\begin{proof}
In order to proof this statement we find compatibility conditions for Lioiville system \eqref{eq:Liouville} for \eqref{eq:cubic_osc}. This is done by excluding $\psi_{2}$ and $\psi_{3}$ from \eqref{eq:Liouville} via
\begin{equation}\label{eq:psi_2_3_def}
  \begin{gathered}
 \psi_{2}=\frac{3\psi_{1,x}+2f\psi_{1}}{6g}, \quad \psi_{3}=\frac{9\psi_{1,xx}+3f\psi_{1,x}+9g\psi_{1,y}-2f^{2}\psi_{1}+12hg\psi_{1}}{18g^{2}},
  \end{gathered}
\end{equation}
and obtaining a system of two linear partial differential equations for $\psi_{1}$, which has the form
\begin{multline}\label{eq:Liouville_3_1}
27 \psi_{1,xxx}+81g \psi_{1,xy} -27 f \psi_{1,xx}-18\left(f^{2}-3 gh -3 g_{y}\right) \psi_{1,x}+\\
\left(8 f^{3}-36 f g h +108 g^{2} k -36 f g_{y}+36 g f_{y}\right) \psi_{1}=0,
\end{multline}
\begin{multline}\label{eq:Liouville_3_2}
27g \psi_{1,xxy} +9 f \psi_{1,xy} g +27 g^{2} \psi_{1,yy}-18\left(g h +3 g_{y}\right) \psi_{1,xx}-\\
3\left(2f g h -18 g^{2} k + 6 f g_{y}-3 g f_{y}\right) \psi_{1,x}-3 \left(2 f^{2}-6 g h +9 g_{y}\right) g \psi_{1,y}+ \\
4\left(f^{2} g h +9 f \,g^{2} k -6 g^{2} h^{2}+3 f^{2} g_{y}-3f g f_{y}+9 g^{2} h_{y}-9 g h g_{y}\right) \psi_{1}=0
\end{multline}
Below we obtain integrability conditions for \eqref{eq:Liouville_3_1}, \eqref{eq:Liouville_3_2} as an overdetermined system of two equations for the function $\psi_{1}$.

First, we exclude the function $k$ from  \eqref{eq:Liouville_3_1}-\eqref{eq:Liouville_3_2} via \eqref{eq:cubic_lc_3}. Then, in order to simplify the computations, we use the definition of $B$ given in \eqref{eq:cubic_notation} to exclude the function $h$ from \eqref{eq:Liouville_3_1}-\eqref{eq:Liouville_3_2}. This yields
\begin{equation}\label{eq:thc2_eq1a}
3 \psi_{1,xxx}+9 g \psi_{1,xy}  -3 f  \psi_{1,xx}-2\left(6 g_{y}+B \right)  \psi_{1,x}=0,\\
\end{equation}
\begin{multline}\label{eq:thc2_eq1b}
81 g \psi_{1,xxy}  +81 g^{2}\psi_{1,yy}+27 f g \psi_{1,xy}  -18 \left(f^{2}-B +6 g_{y}\right) \psi_{1,xx}-9 g \left(15 g_{y}+2 B \right) \psi_{1,y}-\\
3 \left(4 B f +12 f g_{y}+9 g f_{y}\right) \psi_{1,x}-4 \left(2 B^{2}-6 B g_{y}+9 g B_{y}+27 g g_{yy}-36 g_{y}^{2}\right) \psi_{1}=0.
\end{multline}

To obtain integrability conditions of \eqref{eq:thc2_eq1a}-\eqref{eq:thc2_eq1b} we follow the classical approach by Liouville \cite{Dunajski2009}. Cross-differentiating equations \eqref{eq:thc2_eq1a}-\eqref{eq:thc2_eq1b}, we find an expression for $\psi_{1,xyy}$. Then, differentiating $\psi_{1,xyy}$ by $x$ and equation \eqref{eq:thc2_eq1a} by $y$ we obtain the expression for $\psi_{1,yyy}$. Consequently, we have expressions for all third order partial derivatives through the derivatives of lower order. Finally, we compute the difference between $\psi_{1,xyyy}$ and $\psi_{1,yyyx}$ and using all third order derivatives we obtain the correlation
\begin{equation}\label{eq:thc2_eq2}
45gB_{y}\psi_{1,xy}-15fB_{y}\psi_{1,xx}+(9gB_{yy}-2BB_{y}-51g_{y}B_{y})\psi_{1,x}=0.
\end{equation}
Since  $g\neq0$ further computations split into two cases: $B_{y}\neq0$ and $B_{y}=0$.

Suppose that $B_{y}\neq0$. Then we solve \eqref{eq:thc2_eq2} for $\psi_{1,xy}$ and differentiate the former by $x$. As a result, we find an expression for $\psi_{1,yy}$.  If we differentiate \eqref{eq:thc2_eq2} by $y$ and eliminate $\psi_{1,xyy}$, $\psi_{1,xxy}$ and $\psi_{1,xy}$ we obtain
\begin{multline}\label{eq:thc2_eq3}
\psi_{1,x}\bigg[45 g^{2} B_{y} B_{yyy}-54 g^{2} B_{yy}^{2}-9 g B_{y} \left(B -3 g_{y}\right) B_{yy}+\\  \left(4 B^{2}-9 B g_{y}-15 g B_{y}+45 g g_{yy}-9 g_{y}^{2}\right) B_{y}^{2}\bigg]=0
\end{multline}
Thus, further computations split into two cases: $\psi_{1,x}=0$ and $\psi_{1,x}\neq0$.

Suppose that $\psi_{1,x}=0$. Then the equation \eqref{eq:thc2_eq1a} is fulfilled and from \eqref{eq:thc2_eq1b} we find that $\psi_{1}$ is a solution of a linear homogeneous equation
\begin{equation}\label{eq:thc2_eq4}
81g^{2} \psi_{1,yy} -9 g \left(15 g_{y}+2 B \right) \psi_{1,y}-4 \left(2 B^{2}-6 B g_{y}+9 g B_{y}+27 g g_{yy}-36 g_{y}^{2}\right) \psi_{1}=0.
\end{equation}
Notice that if $\psi_{1,x}=0$ from \eqref{eq:psi_2_3_def} it follows that $\psi_{2,x}=\psi_{3,x}=0$. Thus, we see that equations from \eqref{eq:cubic_osc} that satisfy \eqref{eq:cubic_lc_3} are metrisable without any additional conditions on the functions $f$, $g$, $h$ and $k$, if the corresponding metric do not depend on $x$.

Let us continue with the case of $\psi_{1,x}\neq0$. One more differentiation of \eqref{eq:thc2_eq2} by $x$ yields correlation \eqref{eq:cubic_lc_3_3DLA}. The first derivative of \eqref{eq:cubic_lc_3_3DLA} is presented in \eqref{eq:thc1_eq9}. If we assume that $C_{3}\neq0$ in \eqref{eq:thc1_eq9} and eliminate $B_{yyy}$ and $B_{yy}^{2}$ via \eqref{eq:thc1_eq9} and \eqref{eq:cubic_lc_3_3DLA}, respectively, we obtain that \eqref{eq:thc2_eq3} vanish. Further differentiating of \eqref{eq:thc2_eq2} does not lead to any new compatibility conditions.

If we assume that $C_{3}=0$ (see \eqref{eq:thc1_eq6} for the definition) then we obtain that \eqref{eq:cubic_lc_3_3DLA} and \eqref{eq:thc2_eq3} vanish simultaneously if and only if $B_{y}=0$ and this case is considered below.

Suppose that $B_{y}=0$. Then \eqref{eq:thc2_eq2} vanishes and we have the expressions of all third order derivatives of $\psi_{1}$ via lower order derivatives. Cross-differentiation of this third-order derivative does not result in any new compatibility conditions. Consequently, if $B_{y}=0$ an equation from \eqref{eq:cubic_osc} under conditions \eqref{eq:cubic_lc_3} is metrisable. This completes the proof.
\end{proof}

\begin{corollary}
\label{cr:cr3}
Metrics that correspond to \eqref{eq:cubic_eq} when condition \eqref{eq:cubic_lc_3} is fulfilled and with $\psi_{1,x}\neq0$, which are considered in Theorem 3 of \cite{Agapov2024}, are either Darboux-superintegrable or flat. 
\end{corollary}
\begin{proof}
In \cite{Matveev2008} it was demonstrated that if $\dim{\mathfrak{p}(\mathfrak{g})}=3$ then $\mathfrak{g}$ is Darboux-superintegrable, i.e. $\mathfrak{g}$ has 4 linearly independent quadratic, with respect to momenta, first integrals. These metrics were classified by Koenigs \cite{Koenigs}. If $\dim{\mathfrak{p}(\mathfrak{g})}=8$, then the corresponding projective equation is point equivalent to $y_{xx}=0$ and, thus, $\mathfrak{g}$ has a constant curvature. Such metrics has three, linearly independent, linear in momenta first integrals.

Now we demonstrate that autonomous projective structures considered in \cite{Agapov2024} satisfy condition \eqref{eq:cubic_lc_3}. To this aim we use the notation of  \cite{Agapov2024}. The integrability conditions for the cubic oscillator
\begin{equation}\label{eq:cr2m2}
y_{xx}=A_{3}y_{x}^{3}+A_{2}y_{x}^{2}+A_{1}y_{x}+A_{0}.
\end{equation}
presented in formula (2.18) in \cite{Agapov2024} have the form
\begin{equation}\label{eq:cr2m1}
A_{1}=-\frac{3A_{0}}{v}, \quad A_{2}=\frac{3A_{0}}{v^{2}}+\frac{3v_{y}}{2v}-\frac{w_{y}}{2w}, \quad A_{3}=-\frac{A_{0}}{v^{3}}+\frac{1}{2v}\left(\frac{w_{y}}{w}-\frac{v_{y}}{v}\right).
\end{equation}
Here $v$ and $w$, $vw\neq0$ are axillary functions. In \cite{Agapov2024} the metrisability of \eqref{eq:cr2m2} when \eqref{eq:cr2m1} holds is studied. Notice also that in \cite{Agapov2024} it is assumed that $A_{0}A_{3}\neq0$. However, the requirement $A_{3}\neq0$ can be omitted.

Excluding the functions $v$ and $w$ from the first two equations of \eqref{eq:cr2m1} and substituting the results into the third one we obtain
\begin{equation}\label{eq:cr2m3}
A_{3}=\frac{9A_{2}A_{1}A_{0}-2A_{1}^{3}-9A_{0,y}A_{1}+9A_{1,y}A_{0}}{27A_{0}^{2}}.
\end{equation}
Since $A_{0}\neq0$  from \eqref{eq:cr2m3} we get
\begin{equation}\label{eq:cr2m4}
27A_{0}^{2}A_{3}-9A_{2}A_{1}A_{0}+2A_{1}^{3}+9A_{0,y}A_{1}-9A_{1,y}A_{0}=0.
\end{equation}
Now comparting \eqref{eq:cubic_osc} and \eqref{eq:cr2m2} we see that $A_{0}=-g$, $A_{1}=-f$, $A_{2}=-h$ and $A_{3}=-k$. Consequently, \eqref{eq:cr2m1}, \eqref{eq:cr2m3} and \eqref{eq:cr2m4} are equivalent to
\begin{equation}\label{eq:cr2m4a}
27kg^{2}-9hfg+2f^{3}+9gf_{y}-9fg_{y}=0.
\end{equation}
Taking into account definition \eqref{eq:cubic_notation} of the function $A$  we obtain that conditions \eqref{eq:cr2m1} are equivalent to linearizability condition \eqref{eq:cubic_lc_3}.

Consequently, we see that, in fact, in Theorem 3 of \cite{Agapov2024} the metrisability of \eqref{eq:cubic_osc} under condition \eqref{eq:cubic_lc_3} is studied.  Moreover, in \cite{Agapov2024} it is assumed that $\psi_{1,x}\neq0$. Therefore, due to the first part of this corollary, the corresponding metrics are either Darboux-superintegrable or flat. Furthermore, this means that all metrics presented in \cite{Agapov2024} have polynomial first integrals, and the rational first integrals provided in \cite{Agapov2024} are fractions of these polynomials.

Finally, let use remark that a particular case of condition \eqref{eq:cubic_lc_3} for $k=0$ and constant $h$ was considered in \cite{GL2004}. In general, correlation \eqref{eq:cubic_lc_3} can be considered as an integrability conditions for the family of the Abel differential equations $z_{y}=g z^{3}+fz^{2}+hz+k$, $y_{x}=z^{-1}$ (see, \cite{Kamke} 4.10(c)). This completes the proof.
\end{proof}

\begin{corollary}
All metrics that correspond to \eqref{eq:cubic_osc} under condition \eqref{eq:cubic_lc_3} are superintegrable. Their geodesics are formed by the curves
\begin{equation}\label{eq:explicit_geodesics}
  (x,x(y))=\left(x,c_{2}-\int\limits_{0}^{y}\left(\frac{f(\omega)}{3g(\omega)}-\sqrt{\frac{C^{'}(\omega)}{g(\omega)[c_{1}-2C(\xi)]}}\right)d\omega\right),
\end{equation}
where $c_{1}$ and $c_{2}$ are arbitrary constants.
\end{corollary}
\begin{proof}
If $\psi_{1,x}=0$, then from the relations \eqref{eq:psi_2_3_def} it follows that the metric $\mathfrak{g}$ and Hamiltonian $H$ in \eqref{eq:geodesics_H} do not explicitly depend on $x$. Therefore, there exists a linear first integral $L=p_{1}$.

On the other hand, if correlation \eqref{eq:cubic_lc_3} holds, then \eqref{eq:cubic_osc} admits invariant \eqref{eq:cubic_Inv_3}. If we lift this invariant via $y_{x}=u=H_{p_{2}}/H_{p_{1}}$, we obtain a first integral of Hamiltonian system \eqref{eq:geodesics_H} that explicitly depends on $x$. Clearly, this yields that $H$, $L$ and lifted \eqref{eq:cubic_Inv_3} are functionally independent (see also Proposition 1 in \cite{GS2024}).

If $\psi_{1,x}\neq0$ and both conditions \eqref{eq:cubic_lc_3} and \eqref{eq:cubic_lc_3_3DLA} are fulfilled, then $\mathfrak{g}$ is Darboux superintegrable.
In the case of the flat metrics that are described by \eqref{eq:cubic_lc_3_8DLA} it is known \cite{Matveev2011} that they are superintegrable with two linear with respect to momenta first integrals.

Lastly, suppose that $R=c_{1}$ and $\mathcal{R}=c_{2}$, where $c_{1}$ and $c_{2}$ are arbitrary constants. Then, from \eqref{eq:cubic_FI_3} and \eqref{eq:cubic_Inv_3} we find \eqref{eq:explicit_geodesics}. This completes the proof.
\end{proof}

\begin{corollary}
For any equation from \eqref{eq:cubic_osc} whose coefficients satisfy \eqref{eq:cubic_lc_3} there is the following solution of the Liouville system
\begin{multline}\label{eq:LS_solution}
\psi_{1}=(c_{3}+c_{4}C)\left(\frac{g}{C_{y}}\right)^{2/3}, \quad \psi_{2}=(c_{3}+c_{4}C)\frac{f}{3(gC_{y}^{2})^{1/3}}, \\
\psi_{3}=\frac{9c_{4}gC_{y}+f^{2}(c_{3}+c_{4}C)}{18(g^{2}C_{y})^{2/3}}, \quad gC_{y}\neq0,
\end{multline}
where $C$ is given in \eqref{eq:C_sol} and $c_{3}$ and $c_{4}\neq0$ are arbitrary constants.

Moreover, rational first integral \eqref{eq:cubic_FI_3} is a projection of a combination of the Hamiltonian that corresponds to \eqref{eq:LS_solution} and its linear first integral $L=p_{1}$.
\end{corollary}
\begin{proof}
Suppose that in \eqref{eq:thc2_eq1a}-\eqref{eq:thc2_eq1b} $\psi_{1,x}=0$. Then, \eqref{eq:thc2_eq1a} vanish identically and \eqref{eq:thc2_eq1b} is a linear homogeneous second-order equation for $\psi_{1}$. Excluding $B$ via \eqref{eq:C_def}, its general solution can be easily found and is presented in \eqref{eq:LS_solution}. Then, with the help of \eqref{eq:psi_2_3_def} we find the rest of the expressions from \eqref{eq:LS_solution}. One can find that for \eqref{eq:LS_solution} we have that $\Delta=(c_{3}+c_{4}C)(g/C_{y})^{1/3}(c_{4}/2)$. Since $g C_{y}\neq0$ we see that we need to require $c_{4}\neq0$ so that the corresponding metric $\mathfrak{g}$ do not degenerate.

With the help of the expressions \eqref{eq:LS_solution}, one can find the metric tensor $g_{ij}$, its inverse and the Hamiltonian for the geodesic flow
\begin{multline}\label{eq:cr4_eq1}
H=\frac{c_{4}C+c_{3}}{72gC_{y}} \bigg([(2Cf^2 + 9gC_{y})c_{4} + 2f^{2}c_{3}]p_{1}^2 - \\ 12gf (  c_{4}C +c_{3} )p_{1}p_{2} +18 g^{2}( c_{4}C + c_{3})p_{2}^{2}\bigg)
\end{multline}
Clearly, that this Hamiltonian also admits $L=p_{1}$ as a first integral. Then, we can project the expression $H/p_{1}^{2}$ on the $(x,y)$ plane via $y_{x}=u=H_{p_{2}}/H_{p_{1}}$. As a consequence, we get
\begin{equation}\label{eq:cr4_eq2}
R=\frac{9c_{4}^{3}gC_{y}u^{2}}{16(fu+3g)^{2}}+\frac{c_{4}^{3}}{8}C+\frac{c_{3}c_{4}^{2}}{8}.
\end{equation}
Without loos of generality, setting $c_{2}=2^{4/3}$, we obtain \eqref{eq:cubic_FI_3}. This completes the proof.
\end{proof}

Finally, we provide several examples of metrisable equations from \eqref{eq:cubic_osc} that correspond to superintegrable metrics.

\textbf{Example 2.}
Let us consider the following oscillator from \eqref{eq:cubic_eq} that satisfy condition \eqref{eq:cubic_lc_3_2DLA}
\begin{equation}\label{eq:ex2_1}
  y_{xx}+ (81y^6 + 27y^3 - 7)y_{x}^{3} + \frac{243y^6 + 54y^3 - 4}{3y}y_{x}^{2} + 3y(9y^3 + 1)y_{x} + 3y^3=0.
\end{equation}
One can also show that conditions \eqref{eq:cubic_lc_3_3DLA} and \eqref{eq:cubic_lc_3_8DLA} cannot be fulfilled for \eqref{eq:ex2_1}.

With the help of relations \eqref{eq:cubic_FI_3} and \eqref{eq:cubic_Inv_3} we find autonomous
\begin{equation}\label{eq:ex2_2}
 R=y^{-2/3}\left(1-\frac{9y^{2}u^{2}}{[3y(9y^{3}+1)u+9y^{3}]^{2}}\right),
\end{equation}
and non-autonomous
\begin{multline}\label{eq:ex2_3}
\mathcal{R}=x +\frac{3 y}{2 \left(9 u \,y^{3}+3 y^{2}+u \right)^{3}} \bigg(\left(729 y^{9}+81 y^{6}+27 y^{3}+7\right) y \,u^{3}+\\
\left(9 y^{3}+1\right) \left(81 y^{6}-9 y^{3}+2\right) u^{2}+27 y^{5} \left(9 y^{3}-1\right) u +3 y^{4} \left(9 y^{3}-2\right)\bigg).
\end{multline}
first integrals of \eqref{eq:ex2_1}, respectively.

Equation \eqref{eq:ex2_1} is metrisable and the corresponding Hamiltonian for geodesics is
\begin{equation}\label{eq:ex2_4}
H=y^{-4/3}\left(3(9y^{3}+2)y\frac{p_{1}^{2}}{2}-(9y^{3}-1)p_{1}p_{2}+3y^{2}\frac{p_{2}^{2}}{2}\right).
\end{equation}
This Hamiltonian system is superintegrable with 3 functionally independent integrals $H$, $L=p_{1}$ and the lift of \eqref{eq:ex2_3}, which is
\begin{multline}\label{eq:ex2_5}
T=\frac{1}{2}\left(2x -3 y^{2} \left(18 y^{3}+5\right) \left(18 y^{3}+1\right)\right)p_{1}^{3}+3 y \left(162 y^{6}+36 y^{3}+1\right) p_{2}p_{1}^{2}-\\
18 y^{3} \left(9 y^{3}+1\right) p_{2}^{2}p_{1}+18 y^{5} p_{2}^{3}.
\end{multline}
Notice that $\dim{\mathfrak{p(\mathfrak{g})}}=2$ for \eqref{eq:ex2_4} since condition \eqref{eq:cubic_lc_3_2DLA} is met. Projective equation \eqref{eq:ex2_1} provides an example of non-trivial metrisable equation with a rational first integral that satisfy condition \eqref{eq:cubic_lc_3_2DLA}. 

\textbf{Example 1 continued.}
For arbitrary values of $\mu$ and $\nu$ the first integral of  \eqref{eq:ex1_1} is expresses in terms of the gamma function and the relative non-autonomous first integral is difficult to obtain.
Let us discuss several particular cases of \eqref{eq:ex1_1}, when both autonomous and non-autonomous first integrals can be presented in the explicit form.

Suppose, that $\nu=0$ and $\mu=1$. Then, \eqref{eq:ex1_1} possesses the first integral
\begin{equation}\label{eq:ex1_2}
R=y^{2}\left(\frac{u^{2}}{y^{2}(u+y)^{2}}+\frac{2}{y}\right),
\end{equation}
and invariant
\begin{equation}\label{eq:ex1_3}
\mathcal{R}=(y \mathrm{e})^{\frac{\sqrt{(2y + 1)u^2 + 4uy^2 + 2y^3}}{u + y}}\frac{(\sqrt{(2y + 1)u^2 + 4uy^2 + 2y^3} + u)^{2}}{2u^2y + 4uy^2 + 2y^3}.
\end{equation}
Equation \eqref{eq:ex1_1} is metrisable for arbitrary values of $\mu$ and $\nu$. At $\mu=1$ and $\nu=0$ the corresponding Hamiltonian for geodesics is
\begin{equation}\label{eq:ex1_4}
H=y(2y-1)(2p_{1}^2 -2 (2y - 1)p_{1}p_{2} + (2y - 1)yp_{2}^{2}).
\end{equation}
This Hamiltonion has the linear first integral $L=p_{1}$. One can demonstrate that lifted first integral \eqref{eq:ex1_2} is a function of $L$ and $H$, while \eqref{eq:ex1_4} is functionally independent of $L$ and $H$. Therefore, \eqref{eq:ex1_4} is superintegrable with linear and a transcendental first integrals.

\begin{table}[!ht]
  \centering
  \begin{tabular}{llll}
    \hline
     & Parameters & $\dim{\mathfrak{p}(c)}$ \\
    \hline
    1 &\begin{tabular}[l]{l} $\beta_{3}=\alpha_{3}(\beta_{0}+1), \, \beta_{2}=\alpha_{1}(2\beta_{0}+1)/2, \, \beta_{1}=2\beta_{0}$ \\ $(\beta_{0}-1)^2+\alpha_{3}^{2}+\alpha_{1}^{2}\neq0,\, $ \end{tabular} & \begin{tabular}[l]{l} $1,  \alpha_{3}^{2}+ \alpha_{1}^{2}\neq0$,\\ $8,  \alpha_{3}= \alpha_{1}=0 $ \end{tabular}\\
    2 & $\beta_{3}=\beta_{2}=\alpha_{1}=0, \, \beta_{1}=-4,\, \beta_{0}=2$ &  \begin{tabular}[l]{l} $1, \alpha_{3}\neq0$,\\ $8,  \alpha_{3}=0 $ \end{tabular}\\
    3 & $\beta_{3}=\beta_{2}=\alpha_{1}=0, \, \beta_{1}=2,\, \beta_{0}=-1$ &  \begin{tabular}[l]{l} $1, \alpha_{3}\neq0$,\\ $8,  \alpha_{3}=0 $ \end{tabular}\\
    4 & $\beta_{3}=\beta_{2}=\alpha_{1}=0, \, \beta_{1}=-1\pm3\sqrt{3},\, \beta_{0}=-4$, &  \begin{tabular}[l]{l} $1, \alpha_{3}\neq0$,\\ $8,  \alpha_{3}=0 $ \end{tabular}\\
    5 & $\beta_{3}=0, \, \beta_{1}=-5/2,\, \beta_{0}=-1,\, \alpha_{1}=5\beta_{2}/2,\, \alpha_{3}=\beta_{2}^{2}$ &  \begin{tabular}[l]{l} $1, \beta_{2}\neq0$,\\ $8,  \beta_{2}=0 $ \end{tabular}\\
    6 & $\beta_{3}=0, \, \beta_{1}=2,\, \beta_{0}=-1,\, \alpha_{1}=-2\beta_{2}$ & \begin{tabular}[l]{l} $1,  \alpha_{3}^{2}+ \beta_{2}^{2}\neq0$,\\ $8,  \alpha_{3}= \beta_{2}=0 $ \end{tabular}
  \end{tabular}
  \caption{Values of the parameters that correspond to metrisable cases of system \eqref{eq:ex3_1}}\label{t:t1}
\end{table}

\textbf{Example 3}
We consider a variant of the cubic Kolmogorov system
\begin{align}
\begin{split}\label{eq:ex3_1}
  & y_{x} = y(\alpha_{0}-\alpha_{1}y-\alpha_{2}z-\alpha_{3}y^{2}) ,\\
  & z_{x} = z(\beta_{0}-\beta_{1}z-\beta_{2}y-\beta_{3}y^{2}-\beta_{4}z^{2}).
\end{split}
  \end{align}
Here $\alpha_{j}$, $j=0,1,2,3$ and $\beta_{k}$, $k=0,1,2,3,4$ are arbitrary parameters. We assume that $\alpha_{0}\alpha_{2}\beta_{4}\neq0$. Thus, without loss of generality, we set  $\alpha_{0}=\alpha_{2}=\beta_{4}=1$. Then, we can exclude $z$ via the first equation from \eqref{eq:ex3_1} and substitute the result into the second one. As a consequence, we obtain an oscillator from \eqref{eq:cubic_osc}, which is
\begin{equation}\label{eq:ex3_2}
\begin{gathered}
   y_{xx}+\frac{y_{x}^{3}}{y^{2}}+\left(3a_{3}y+3a_{1}-\frac{b_{1}+4}{y}\right)y_{x}^{2}+\bigg(3a_{3}^2y^4 + 6a_{1}a_{3}y^3+\\
   +(3a_{1}^2 - 2a_{3}b_{1} - 4a_{3} + b_{3})y^2-(2a_{1}b_{1}+5a_{1} - b_{2})y - b_{0} + 2b_{1} + 3\bigg)y_{x}+\\
   +y \bigg(a_{3}y^2 + a_{1}y - 1\bigg)\bigg(a_{3}^2y^4 + 2a_{3}a_{1}y^3 + (a_{1}^2 - a_{3}b_{1} - 2a_{3} + b_{3})y^2 +\\+ ( b_{2}-a_{1}b_{1} - 2a_{1})y - b_{0} + b_{1} + 1\bigg)=0.
  \end{gathered}
\end{equation}

Now we substitute the coefficients of \eqref{eq:ex3_2} into \eqref{eq:cubic_lc_3} and find the values of the parameters when this correlation is fulfilled. The results are presented in Table \ref{t:t1}. Notice that we do not consider complex values of the parameters and do not show cases of equations with $\dim \mathfrak{p}(\mathfrak{g})=8$ for any values of the parameters. For any value of the parameters an equation from \eqref{eq:ex3_2} will be metrisable and have a rational first integral. The existence of this rational first integral follows from integrability of the corresponding geodesic flow with a linear first integral $L=p_{1}$. 

However, metrisable equations from \eqref{eq:ex3_2} also have an invariant given by \eqref{eq:cubic_Inv_3} and, hence, explicit representation for geodesics \eqref{eq:explicit_geodesics}. For example, let us consider the case 8 from Table \ref{t:t1}. For a sake of simplicity, we also assume that $\beta_{2}=0$ and $\alpha_{3}=-\gamma^{2}$. As a result, we find that the geodesic flow
\begin{equation}\label{eq:ex3_3}
\begin{gathered}
H=-\gamma^{2}\left(\frac{p_{1}^{2}}{2}-\gamma^{2}y^{3}p_{1}p_{2}+\frac{\gamma^{2}y^{4}(A^{2}y^{2}-1)}{2}p_{2}^{2}\right),
  \end{gathered}
\end{equation}
is superintegrable with two additional firs integrals
\begin{equation}\label{eq:ex3_4}
\begin{gathered}
L=p_{1},\\
T=\frac{\left(2 p_{1}-\left(\gamma^{2} y^{3}-y \right) p_{2}\right)}{y^{3} p_{2}} {\mathrm e}^{2 x -\frac{2 p_{1} \left(\gamma^{2} y^{3} p_{2}-y p_{2}-p_{1}\right)}{y^{3} \gamma^{2} \left(\gamma^{2} y^{3} p_{2}-y p_{2}-2 p_{1}\right) p_{2}}}.
  \end{gathered}
\end{equation}
There are other examples of explicitly superintegrable geodesic flows among cases listed in Table \ref{t:t1}.

\section{Integrable conformal metrics}
\label{sec:conf}
In the previous Section the metrisability  of equations from \eqref{eq:cubic_osc} than can be linearized via \eqref{eq:GNT} is studied only in one case, namely when correlation \eqref{eq:cubic_lc_3} holds. Finding metrisable equations that satisfy either condition \eqref{eq:cubic_lc_1} or \eqref{eq:cubic_lc_2} for a generic metric $\mathfrak{g}$ is computationally difficult. However, this can be overcome for special cases of the metric $\mathfrak{g}$. 

First, we assume that $g_{12}=g_{21}=0$, which implies that $\psi_{2}=0$. Then we in addition suppose that $g_{11}=g_{22}$, which yields $\psi_{1}=\psi_{3}$. The former case of $\mathfrak{g}$  correspond to the orthogonal metrics, while the latter one is called conformal metrics \cite{Bolsinov1998}. 

\begin{theorem}
\label{th:th5}
Equation \eqref{eq:cubic_osc} is metrisable in the case of $\psi_{2}=0$ if and only if either
  \begin{equation}\label{eq:mc_p2_0_1}
     f k_{y}=-2gk^{2}+2fhk, \quad f_{y}=0, \quad  fk\neq0
  \end{equation}
  or
    \begin{equation}\label{eq:mc_p2_0_2}
    k=f=0.
  \end{equation}
The solution of the Liouville system \eqref{eq:Liouville} is either
  \begin{equation}\label{eq:mc_p2_0_1_LS}
     \psi_{1}=\frac{\epsilon}{k}\psi_{3}, \quad \psi_{3}=c_{5}\exp\left\{\frac{2}{3}\left(\int h dy -\epsilon x\right)\right\}, \quad f=\epsilon\neq0, \quad c_{5}k\neq0,
  \end{equation}
or
  \begin{equation}\label{eq:mc_p2_0_2_LS}
     \psi_{1}={\rm e}^{-\frac{4}{3}\int h dy}\left(2c_{6}\int g {\rm e}^{2\int h dy} dy +c_{7}\right), \quad \psi_{3}=c_{6}{\rm e}^{\frac{2}{3}\int h dy},
  \end{equation}
respectively. Here $c_{5}$, $c_{6}$ and $c_{7}$ are arbitrary constants such that $c_{5}c_{6}\neq0$.

In the case of the conformal metric $\mathfrak{g}$, i.e. $\psi_{1}=\psi_{3}$ and $\psi_{2}=0$, we have that \eqref{eq:cubic_osc} is metrisable if and only if
  \begin{equation}\label{eq:mc_p2_0_3}
     k=f, \quad h=g, \quad f_{y}=0.
  \end{equation}
The solution of the Liouville system is
  \begin{equation}\label{eq:mc_p2_0_3_LS}
     \psi_{1}=\psi_{3}=c_{5}\exp\left\{\frac{2}{3}\left(\int g dy -\epsilon x\right)\right\}, \quad f=k=\epsilon\neq0, \quad c_{5}\neq0.
  \end{equation}
\end{theorem}
\begin{proof}
First, we substitute that $a_{3}=k$, $a_{2}=h$, $a_{1}=f$, $a_{0}=g$ and $\psi_{2}=0$ into \eqref{eq:Liouville}. As a result, we have four equations for two functions $\psi_{1}\neq0$ and $\psi_{3}\neq0$, which are
\begin{align}\label{eq:th4_1}
  \begin{split}
     3\psi_{1,x}+2f\psi_{1}=&0 , \quad 3\psi_{3,y}-2h\psi_{3}=0, \\
     3\psi_{1,y}+4h\psi_{1}-6g\psi_{3}=&0, \quad   3\psi_{3,x}+6k\psi_{1}-4f\psi_{3}=0.
  \end{split}
\end{align}
Cross differentiating the first and the third and the second and the fourth equations from \eqref{eq:th4_1} we get
  \begin{equation}\label{eq:th4_2}
     12 fg \psi_{3}-12gk\psi_{1}+2f_{y}\psi_{1}=0, \quad 12hk\psi_{1}-12kg\psi_{3}-6k_{y}\psi_{1}+4f_{y}\psi_{3}=0.
\end{equation}
Recall, that throughout this work we assume that $g\neq0$. Then, from the first equation of \eqref{eq:th4_2} we see that we need to consider two cases separately: $f\neq0$ and $f=0$.

Suppose that $f\neq0$. Then we exclude $\psi_{3}$ from the first equation from \eqref{eq:th4_1} and then substitute the result into the second one. Recall that if $\psi_{2}=0$, then $\Delta=\psi_{1}\psi_{3}\neq0$. As a consequence, taking into account that $\psi_{1}fg\neq0$, we get
  \begin{equation}\label{eq:th4_3}
    \psi_{3}=\frac{(6gk-f_{y})\psi_{1}}{6fg}, \quad 9g(2hkf - 2gk^2 + kf_{y} - fk_{y})-f_{y}^{2}=0.
\end{equation}
Then, substituting the expression for $\psi_{3}$ along with the expression for $\psi_{1,x}$, obtained from the first equation of \eqref{eq:th4_1}, to the last equation of \eqref{eq:th4_1} we obtain that $f_{y}=0$.  As a result, we immediately get compatibility conditions \eqref{eq:mc_p2_0_1} and the expression for $\psi_{3}$ from \eqref{eq:mc_p2_0_1_LS}.

Suppose that $f=0$. Then, from the first equation of \eqref{eq:th4_2} we find that $k=0$ and the second one is satisfied automatically. Substituting $f=k=0$ into \eqref{eq:th4_1} we find the expressions for $\psi_{1}$ and $\psi_{3}$ given in \eqref{eq:mc_p2_0_2_LS}.

Finally, we consider the case of the conformal metrics, i.e. we assume that $\psi_{3}=\psi_{1}\neq0$. As a consequence, from \eqref{eq:th4_1} we obtain that $g=h$ and $f=k$. Then, comparing $\psi_{1,xy}$ and $\psi_{1,yx}$ we obtain that $f_{y}=0$. The expression for $\psi_{1}$ can be directly found from the fist line of \eqref{eq:th4_1}. Notice also that the above metrics have infinitesimal homotety $\partial_{x}$. This completes the proof. 
\end{proof}

Now let us find intersections between families of linearizable equations provided in Proposition \ref{p:p2} and metrisable oscillators given in Theorem \ref{th:th5}. First, we consider the case of conformal metric $\mathfrak{g}$ and linearizability condition \eqref{eq:cubic_lc_2}:

\begin{theorem}
\label{th:th6}
The metric $\mathfrak{g}$ of the form
  \begin{equation}\label{eq:th6_1}
    ds^{2}=\exp\left\{2\left( x-\int g dy\right)\right\}\left(dx^{2}+dy^{2}\right),
  \end{equation}
with the geodesic flow with the Hamiltonian
\begin{equation}\label{eq:th6_2}
 H=\frac{1}{2}\exp\left\{2\left(\int g dy - x\right)\right\}\left(p_{1}^{2}+p_{2}^{2}\right),
\end{equation}
is integrable with the transcendental first integral
\begin{equation}\label{eq:th6_3a}
T=K_{1}^{1-i\sqrt{3}}K_{2}^{1+i\sqrt{3}}K_{3}^{-2}K_{4},
\end{equation}
or in the real form
\begin{equation}\label{eq:th6_3}
 T=\exp\left\{2\sqrt{3}\arctan\left(\frac{\sqrt{3}((l-1)p_{2}-3gp_{1})}{K_{3}}\right)\right\}\frac{K_{5}}{K_{3}^{2}},
\end{equation}
if and only if the function $g\neq0$ is a solution of the equation
  \begin{equation}\label{eq:th6_4}
 BA-gA_{y}=9 g g_{yy}-27 g_{y}^{2}-\left(9 g^{2}-15\right) g_{y}+2 \left(9 g^{2}+1\right) \left(3 g^{2}-1\right)=0.
\end{equation}

Here we use the notations 
  \begin{equation}\label{eq:th6_5}
A=18g^2 + 2 - 9g_{y}, \quad B=3g_{y}+3g^{2}-1, \quad l^{3}=A,
\end{equation}
and $K_{j}$, $j=1,2,3,4$ are the relative Killing vectors that given by
\begin{equation}
\begin{gathered}
\label{eq:th6_5a}
  K_{1,2}=12gp_{1} + (1 \mp \sqrt{3}i)(1 - 2l \pm \sqrt{3}i )p_{2}, \quad K_{3}=3gp_{1} + (l + 1)p_{2}, \\
   K_{4}=\exp\left\{\frac{2}{3} \int \frac{l^{2}}{g} dy\right\}, \quad K_{5}=K_{1}K_{2}=9g^2p_{1}^2 - 3(l - 2)gp_{1}p_{2} + 16(l^2 - l + 1)p_{2}^2.
\end{gathered}
\end{equation}

\end{theorem}
\begin{proof}
From Theorem \ref{th:th5} it follows that an equation from \eqref{eq:cubic_osc} is metrisable in the case of the conformal metric if and only if it has the form
\begin{equation}\label{eq:th6_7}
  y_{xx}+ \epsilon y_{x}^{3}+g y_{x}^{2}+ \epsilon y_{x}+g=0, \quad \epsilon g\neq0.
\end{equation}
Notice that, without loss of generality, we can set $\epsilon=1$.

Then, we apply to \eqref{eq:th6_7} integrability condition \eqref{eq:cubic_lc_2} from Proposition \ref{p:p2}. As a result, we obtain that \eqref{eq:th6_7} can be transformed to linear equation \eqref{eq:linear_eq} with \eqref{eq:cubic_lc_2_p} via \eqref{eq:GNT} if and only if correlation \eqref{eq:th6_4} is fulfilled. From Proposition \ref{p:p2} it also follows that linearizability is equivalent to the existence of the first integral \eqref{eq:lc_FI}. Under \eqref{eq:cubic_lc_2_p} the expression for \eqref{eq:lc_FI} is
\begin{multline}\label{eq:th6_8a}
R=\exp \left\{\frac{2}{3} \left(\int \frac{l^{2}}{g}d y \right)\right\}\left[(1 + \sqrt{3}i)(1-\sqrt{3}i - 2l)u + 12g\right]^{1 + \sqrt{3}i} \times \\
\left[(1 - \sqrt{3}i)(1+\sqrt{3}i - 2l)u + 12g\right]^{1 - \sqrt{3}i} \left[(1+l)y_{x}+3g\right]^{2}.
\end{multline}
Then we lift \eqref{eq:th6_8a} to \eqref{eq:th6_3a} via $u=y_{x}=H_{p_{2}}/H_{p_{1}}$. 

The real variant of integral \eqref{eq:th6_8a} has the form
\begin{multline}\label{eq:th6_8}
 R=\exp \left\{\frac{2}{3} \left(\int \frac{l^{2}}{g}d y \right)-2 \sqrt{3}\, \arctan \left(\frac{\sqrt{3}\, \left(\epsilon  y_{x}-l y_{x}+3 g \right)}{l y_{x}+\epsilon  y_{x}+3 g}\right)\right\} \times \\\frac{ \left(9 g^{2}+3 y_{x} \left(2 \epsilon -l \right) g +y_{x}^{2} \left(\epsilon^{2}-\epsilon  l +l^{2}\right)\right)}{\left(\left(l +\epsilon \right) y_{x}+3 g \right)^{2}}
\end{multline}
In the same way it can be lifted into \eqref{eq:th6_3}. 

One cans see that the first integral \eqref{eq:th6_8a} is formed by the invariants of \eqref{eq:cubic_osc}, that are linear in the first derivative. Consequently, using Proposition \ref{p:p_KV} one can find the corresponding relative Killing vectors. This completes the proof.
\end{proof}

Now we proceed with the most general linearizability criterion from Proposition \ref{p:p2}.

\begin{theorem}
\label{th:th7}
The metric $\mathfrak{g}$ of the form
  \begin{equation}\label{eq:th7_1}
    ds^{2}=\exp\left\{2\left( x-\int g dy\right)\right\}\left(dx^{2}+dy^{2}\right),
  \end{equation}
with the geodesic flow with the Hamiltonian
\begin{equation}\label{eq:th7_2}
 H=\frac{1}{2}\exp\left\{2\left(\int g dy -x\right)\right\}\left(p_{1}^{2}+p_{2}^{2}\right),
\end{equation}
is integrable with the first integral
\begin{equation}\label{eq:th7_3}
 T=K_{6}^{\rho+\beta}K_{7}^{\rho-\beta}K_{8}^{-2\rho}K_{9},
\end{equation}
if and only if the function $g\neq0$ is a solution of the equation
\begin{equation}\label{eq:th7_4}
 \beta^2(2\beta^2 - 9\alpha)^2(g A_{y} - BA)^3 - A^{5}(\beta^2 - 3\alpha)^{3}=0,
\end{equation}
such that $(2\beta^2 - 9\alpha)(\beta^2 - 3\alpha)A\neq0$.
Here the functions $A$, $B$ and $l$ are provided in \eqref{eq:th6_5} and
\begin{equation}\label{eq:th7_5}
\rho=\sqrt{\beta^{2}-4\alpha}, \quad \delta=2\beta^2 - 9\alpha.
\end{equation}
If, in addition, the correlation
\begin{equation}\label{eq:th7_6}
  4\alpha=\beta^{2}(1-r^{2}), \quad r \in \mathbb{Q}, \quad r\neq\pm1, \quad r\neq\pm\frac{1}{3},
\end{equation}
holds, then metric \eqref{eq:th7_1} is integrable with a rational first integral of an arbitrary degree given by
\begin{equation}\label{eq:th7_7}
T_{r}= K_{1}^{p+q}K_{2}^{p-q}K_{3}^{-2p}K_{4}^{\frac{q}{\beta}}.
\end{equation}
The relative Killing tensors $K_{j}$, $j=6,7,8,9$ have the form
\begin{equation}
\begin{gathered}
\label{eq:th7_7a}
 K_{1,2}= \left(3 \rho \mp \beta \right)l p_{2}\mp 2 (\delta\beta)^{\frac{1}{3}} \left(3 g p_{1}+\epsilon  p_{2}\right), \quad
 K_{3}=\beta^{\frac{2}{3}}lp_{2} - \delta^{\frac{1}{3}}(3gp_{1} + \epsilon p_{2}), \\
    K_{4}=\exp\left\{2 \alpha \rho  (\delta\beta)^{-\frac{2}{3}} \int l^{2}g^{-1}d y \right\}. 
\end{gathered}
\end{equation}

\end{theorem}
\begin{proof}
The proof is similar to those of Theorem \ref{th:th6}. First, we observe that the projective equation that correspond to \eqref{eq:th7_2} is \eqref{eq:th6_7} and set $\epsilon=1$. Then, we apply integrability criterion \eqref{eq:cubic_lc_3} to \eqref{eq:th6_7}. Consequently, we get that \eqref{eq:th6_7} is linearizable and has first integral \eqref{eq:lc_FI} if and only if \eqref{eq:th7_4} is fulfilled. Then, after some simplifications, from \eqref{eq:lc_FI} we obtain 
\begin{multline}\label{eq:th7_8}
T=\exp\left\{\frac{2 \alpha \rho \int \frac{l^{2}}{g}d y }{ (\delta\beta)^{\frac{2}{3}}}\right\}\left[ \beta^{\frac{2}{3}} l y_{x} -\delta^{\frac{1}{3}} \left(\epsilon  y_{x}+3 g \right)\right]^{-2\rho}\times \\ \left[  \left(\beta +3 \rho \right)l y_{x}+2 (\beta\delta)^{\frac{1}{3}} \left(  y_{x}+3 g \right) \right]^{\rho -\beta} 
\left[\left(3 r -\beta \right)l y_{x} -2 (\beta\delta)^{\frac{1}{3}} \left( y_{x}+3 g \right) \right]^{\rho+\beta}.
\end{multline}
This integral can be lifted to \eqref{eq:th7_3}. 

Finally, in order to obtain a rational first integral from \eqref{eq:th7_3} we substitute \eqref{eq:th7_6} into it. Notice that we exclude the cases of $r^{2}=1$ and $r^{2}=1/9$ since in the former $\alpha=0$ and in the latter $2\beta^{2}-9\alpha=0$. 

The relative Killing tensors $K_{j}$, $j=6,7,8,9$ can be obtained via formulas \eqref{eq:Killing_vector} from Proposition \ref{p:p_KV} from the multipliers of \eqref{eq:th7_8}, which are invariants linear in $u$ of the vector field associated with \eqref{eq:cubic_osc}. This completes the proof.
\end{proof}

Notice that we do not provide explicit expressions for the cofactors of the relative Killing vectors given in Theorems \ref{th:th6} and \ref{th:th7} since they have a cumbersome form.

Let us connect and generalize the results of Theorems \ref{th:th6} and \ref{th:th7} by introducing the concept of generalized Darboux functions and generalized Darboux integrability for both projective equation \eqref{eq:cubic_eq} and Hamiltonian system \eqref{eq:geodesics_H}. Recall that, in the classical theory of integrability of polynomial finite-dimensional vector fields, the Darboux functions are defined as $f_{1}^{s_{1}}\ldots f_{r}^{s_{r}} \exp\left\{\frac{f_{r+1}}{f_{r+2}}\right\}$, where $g$, $f$, $f_{1}^{s_{1}}\ldots f_{r}^{s_{r}}\in\mathbb{C}[x]$, $s_{1},\ldots,s_{r}\in\mathbb{C}$ and $x\in\mathbb{R}^{n}$ \cite{Llibre_book,Zhang,Goriely}. A finite-dimensional vector field has a first integral which is a Darboux function, if and only if the multipliers of this Darboux function are polynomial and/or exponential invariants of this vector field, whose cofactors are linearly dependent \cite{Llibre_book,Zhang,Goriely}. It is said that a polynomial $n$-dimensional vector field is Darboux integrable if it has $n-1$ functionally independent first integrals that are Darboux functions. 

On the other hand, one can see that the vector fields $\mathcal{X}_{\Gamma}$ and $\mathcal{X}_{H}$ are polynomial only in the velocities and momenta, respectively, while both vector fields $\mathcal{X}_{p}$ and $\mathcal{X}$ are polynomial only in the first derivative $u = y_{x}$. Furthermore, from Proposition \ref{p:p2} and Theorems \ref{th:th6} and \ref{th:th7}, we see that these vector fields can have first integrals which are Darboux functions of invariants that are polynomial only in the above subsets of the phase variables. Thus, it is natural to relax the requirement that both the vector fields and the invariants are polynomials and assume that they are polynomial only in a subset of the phase variables. A similar concept for two-dimensional vector fields was considered in \cite{GG2003,GG2010,GS2025}. This allows us to introduce the definitions of a \textit{generalized Darboux function} and \textit{generalized Darboux integrability} for both projective and Hamiltonian geodesic vector fields. 

First, we consider the projective vector fields $\mathcal{X}_{p}$ and their invariants $R$ that belong to $C^{\infty}(M)[u]$ with the cofactors $l$ also belonging to $C^{\infty}(M)[u]$. We will call the Darboux functions of such invariants generalized Darboux functions. Moreover, we say that the vector field $\mathcal{X}_{p}$  is generalized Darboux integrable if it has two functionally independent first integrals which are generalized Darboux functions. It is a straightforward computation to verify that the vector field $\mathcal{X}_{p}$ is generalized Darboux integrable if and only if the cofactors of the multipliers of the corresponding generalized Darboux function are linearly dependent (see, e.g. \cite{GS2025}). 

Second, in a similar way, we can introduce the definition of a generalized Darboux function and generalized Darboux integrability for Hamiltonian system \eqref{eq:geodesics_H} using relative Killing tensors. Indeed, by definition, $K$ and the respective cofactors belong to $C^{\infty}(M)[p_{1},p_{2}]$ (see, \cite{Kruglikov2024,Kruglikov2025,Kruglikov2026}). We will call the Darboux functions of the relative Killing tensors generalized Darboux functions. We will call a Hamiltonian system integrable in the generalized Darboux sense if it has a first integrals that is a generalized Darboux function, is in involution with the Hamiltonian and is functionally independent of it. 

One can see that both the linearizable cases of \eqref{eq:cubic_osc} and Hamiltonian \eqref{eq:th7_2} under either condition \eqref{eq:th6_4} or \eqref{eq:th7_4} provide examples of generalized Darboux integrable finite-dimensional vector fields and Hamiltonian systems, correspondingly. Further details and algorithms for the construction of generalized Darboux first integrals for both the projective equation and Hamiltonian geodesic flow will be developed elsewhere. Note also that in \cite{GS2025} this type of integrals and invariants was systematically considered for two-dimensional vector fields.

\begin{proposition}
  Conditions \eqref{eq:th6_4} and \eqref{eq:th7_4} are integrable as equations for the function $g$.
\end{proposition}
\begin{proof}
We begin with condition \eqref{eq:th6_4}. If we integrate it as an equation for $A$ we obtain
\begin{equation}\label{eq:p7_1}
R=Ag^{-3}\exp\left\{\int \left(\frac{\epsilon^{2}}{g}-3g\right)dy\right\}.
\end{equation}
If we differentiate \eqref{eq:p7_1} with respect to $y$, we obtain \eqref{eq:th6_4}. 

On the other hand, we can directly consider integrability of \eqref{eq:th6_4} as an equation for $g$. One can demonstrate, for example using the approach from \cite{GS2025}, that \eqref{eq:th6_4} as a dynamical system has 3 first order with respect to $g_{y}$ invariants curves. Taking into account that $g$ itself is also an invariant curve we obtain the following integrating factor for \eqref{eq:th6_4}
  \begin{equation}\label{eq:p7_2}
M=\frac{g}{\left(9 g_{y}^{2} \epsilon^{2}-6 \epsilon^{2} \left(\epsilon^{2}-3 g^{2}\right) g_{y} +\left(\epsilon^{2}+9 g^{2}\right) \left(\epsilon^{2}-3 g^{2}\right)^{2}\right)
 \left(9 g_{y}-2 \epsilon^{2}-18 g^{2} \right)^{\frac{1}{3}}}.
\end{equation}
Notice that these invariant curves are in the denominator of \eqref{eq:p7_2}. 

Now we proceed with the condition \eqref{eq:th7_4}. This is a second-order third-degree differential equation for $g$, i.e. it is cubic in $g_{yy}$. Thus, \eqref{eq:th7_4} can be considered as a multiplication of three second-order first-degree differential equations for $g_{yy}$. Each of these three equations has a first integral, which can be presented in a compact form via the function $C$ as follows
\begin{equation}\label{eq:p7_3}
R_{1}=gA^{-2/3}C_{y}+\frac{2(\beta^{2}-3\alpha)C}{3\beta^{2/3}(2\beta^{2}-9\alpha)^{2/3}}, \quad R_{2,3}=gA^{-2/3}C_{y}+\frac{(\pm i\sqrt{3}-1)(\beta^{2}-3\alpha)C}{3\beta^{2/3}(2\beta^{2}-9\alpha)^{2/3}}
\end{equation}
This completes the proof.
\end{proof}

\begin{proposition}
  In the generic cases $\dim{\mathfrak{p}(\mathfrak{g})}=1$ for metrics \eqref{eq:th6_1} defined by \eqref{eq:th6_4} and \eqref{eq:th7_4}.
\end{proposition}
\begin{proof}
  The proof is straightforward. First, observe that for \eqref{eq:cubic_osc} the projective Lie algebra is at least one-dimensional. Second, we compute relative differential invariants $J_{j}$, $j=0,1,2,3,4$ for equations from \eqref{eq:cubic_osc} of the form \eqref{eq:mc_p2_0_3} and that satisfy either \eqref{eq:th6_4} or \eqref{eq:th7_4}. Third, we find that $J_{0}\neq0$ and basis algebraic invariants $I_{j}$, $j=1,2,3,4$  are not constant. Finally, according to \cite{Bagderina2016} the projective Lie algebra is two=dimensional if and only if all basis invariants are constant. This completes the proof.
\end{proof}

\section{Conclusion}
In this work we have considered metrisability of the autonomous case of the projective connection. We have constructed three families of (super)integrable two-dimensional metrics that are parameterised by arbitrary functions. In the first case, we have obtained a family of superintegrable metrics with an explicit expression for the unparameterized geodesics. In the other two cases we have constructed families of integrable metrics with transcendental first integrals, which at some values of the parameters degenerate into a rational one with an arbitrary degree. We have studied the dimensions of the projective Lie algebras of obtained metrics. We show that all constructed metrics generically have one-dimensional projective Lie algebra. We have found particular instances, when this Lie algebra is 2, 3 and 8 dimensional. This has allowed us to demonstrate when obtained metrics reduced to known or flat ones. The developed mechanism for the constructing integrable metrics may be further extended by considering projective equations \eqref{eq:cubic_eq}, including its autonomous case \eqref{eq:cubic_osc}, which can be linearised by non-autonomous transformations \eqref{eq:GNT}. Moreover, we have introduced the notion of the generalized Darboux integrability in the context of both projective vector fields and geodesic flows and demonstrate that both families of metrics obtained in Section \ref{sec:conf} are integrable in this sense. We believe that the concept of the generalized Darboux integrability can be useful for constructing and classifying integrable geodesic flows with the help of the relative Killing vectors and/or the invariants of the projective vector fields.

\section{Acknowledgments}
J.G. is partially supported by the Agencia Estatal de Investigac\'ion grant PID2020-113758GB-I00 and AGAUR grant number 2021SGR 01618. D.S. is partially supported by H2020-MSCA-COFUND-2020-101034228-WOLFRAM2. Authors are grateful to anonymous referees for their valuable comments and suggestions.

\end{document}